\documentclass{llncs}

\usepackage[hidelinks]{hyperref}

\usepackage[T1]{fontenc}

\usepackage{graphicx}
\usepackage{enumerate}
\usepackage{amsfonts}
\usepackage{booktabs}
\usepackage{siunitx}
\usepackage{multirow}
\usepackage{graphicx}
\usepackage{listings}
\usepackage{commath}
\usepackage{epstopdf}
\usepackage{url}
\usepackage{mathtools}
\usepackage{alltt}
\usepackage{mathrsfs}
\usepackage{wrapfig}

\usepackage{subcaption}
\usepackage[normalem]{ulem}
\usepackage{etoolbox}
\usepackage{stmaryrd}
\usepackage[noend]{algorithmic}

\usepackage{mathtools}
\usepackage{etoolbox}
\usepackage{dsfont}
\usepackage{color}
\usepackage{colortbl}
\usepackage{multirow}
\usepackage[scaled]{helvet}

\usepackage{fullminipage}

\usepackage{graphicx,fancyvrb}
\usepackage[colorinlistoftodos]{todonotes}
\usepackage[noend]{algorithmic}
\usepackage{enumitem}
\usepackage{booktabs}
\usepackage{capt-of}
\usepackage{xcolor}
\usepackage{multirow}

\usepackage{float}
\usepackage{subfloat}

\usepackage{array}
\usepackage{arydshln}
\usepackage{todonotes}

\usepackage[linesnumbered,lined,ruled,vlined]{algorithm2e}

\definecolor{mygreen}{rgb}{0,0.6,0}
\definecolor{mygray}{rgb}{0.5,0.5,0.5}
\definecolor{mymauve}{rgb}{0.58,0,0.82}

\usepackage{listings}
\lstnewenvironment{VerbatimText}[1][]{
    
    \lstset{fancyvrb=true,basicstyle=\footnotesize,captionpos=b,xleftmargin=2em,#1}
}{}

\newcommand{\ignore}[1]{}

\newcommand{\F}{\mathbb F} 

\newcommand{\batya}[1]{{\texttt{\color{blue} Batya: [{#1}]}}}

\newcommand*{\rom}[1]{\expandafter\@slowromancap\romannumeral #1@}
\newcommand{\RNum}[1]{\uppercase\expandafter{\romannumeral #1\relax}}

\newtheorem{defn}{Definition}[section]

\newcommand{\eat}[1]{}

\eat{

\newtheorem{corollary}[defn]{Corollary}

\newtheorem{proposition}[defn]{Proposition}

\newcommand{\sel}[1]{{\sigma}}

}

\newcommand{\cut}[1]{}

\def\set#1{\mathord{\{#1\}}}

\def\eqdef{\mathrel{\stackrel{\textsf{\tiny def}}{=}}}

\def\e#1{\emph{#1}}
\newenvironment{citedtheorem}[1]
{\begin{theorem}{\it\e{(#1)}}\,\,}
	{\end{theorem}}
\newenvironment{citedlemma}[1]
{\begin{lemma}{\it\e{(#1)}}\,\,}
	{\end{lemma}}
\newenvironment{citeddefinition}[1]
{\begin{definition}{\it\e{(#1)}}\,\,}
	{\end{definition}}

\newenvironment{repeatresult}[2]
{\vskip0.5em\par\textsc{#1} #2.\em}
{\vskip1em}

\newenvironment{reptheorem}[1]{\begin{repeatresult}{Theorem}{#1}}{\end{repeatresult}}
\newenvironment{replemma}[1]{\begin{repeatresult}{Lemma}{#1}}{\end{repeatresult}}
\newenvironment{repcorollary}[1]{\begin{repeatresult}{Corollary}{#1}}{\end{repeatresult}}

\def\appendix{\par
	\section*{APPENDIX}
	\setcounter{section}{0}
	\setcounter{subsection}{0}
	\def\thesection{\Alph{section}} }

\def\eqdef{\mathrel{\stackrel{\textsf{\tiny def}}{=}}}

\def\e#1{\emph{#1}}

\newcommand{\algname}[1]{{\sf #1}}
\def\myrulewidth{2.80in}
\def\therule{\rule{\myrulewidth}{0.2pt}}

\def\myrulewidthwide{4in}
\def\therulewide{\rule{\myrulewidthwide}{0.2pt}}

\newenvironment{algserieswide}[2]
{\centering\begin{figure}[#1]\begin{center}\def\thecaption{\caption{#2}}
			\begin{tabular}{p{\myrulewidthwide}}\therulewide\end{tabular}\vskip0.2em}
		{\thecaption\end{center}\end{figure}}

\newenvironment{insidealgwide}[2]
{\normalsize\begin{insidecodewide}{#1}{#2}{Algorithm}}
	{\end{insidecodewide}}

\newenvironment{insidecode}[3]
{
	\begin{tabular}{p{\myrulewidth}}
		\multicolumn{1}{c}{\rule{0mm}{3mm}{\bf #3} $\algname{#1}(\mbox{#2})$\vspace{-0.6em}}\\
		\therule\vskip-0.8em\therule
		\vspace{-1em}
		\begin{algorithmic}[1]}
		{\end{algorithmic}
		\vskip-0.4em\therule
\end{tabular}}

\newenvironment{insidecodewide}[3]
{
	\begin{tabular}{p{\myrulewidthwide}}
		\multicolumn{1}{c}{\rule{0mm}{3mm}{\bf #3} $\algname{#1}(\mbox{#2})$\vspace{-0.6em}}\\
		\therulewide\vskip-0.8em\therulewide
		\vspace{-1em}
		\begin{algorithmic}[1]}
		{\end{algorithmic}
		\vskip-0.3em\therulewide
\end{tabular}}

\newcommand{\minlsep}[2]{\mathcal{S}_{#1}(#2)}

\newcommand{\minlsepst}[1]{\mathcal{S}_{s,t}(#1)}
\newcommand{\minsepst}[1]{\mathcal{L}_{s,t}(#1)}

\def\minsep{\mathcal{L}}

\def\sat{\mathrm{Sat}}

\newcommand{\edges}{\texttt{E}}

\newcommand{\nodes}{\texttt{V}}

\definecolor{mygreen}{rgb}{0,0.6,0}
\definecolor{mygray}{rgb}{0.5,0.5,0.5}
\definecolor{mymauve}{rgb}{0.58,0,0.82}
\definecolor{cadmiumgreen}{rgb}{0.0, 0.42, 0.24}

\def\gq1{{\geq}1}

\def\cc{\mathcal{C}}

\newcommand{\minstVertices}[1]{\mathcal{U}^{s,t}_{min}(#1)}

\newif \ifnonplanar
\nonplanarfalse 

\newcommand{\sminus}{\scalebox{0.75}[1.0]{$-$}}
\newcommand{\splus}{\scalebox{0.75}[1.0]{$+$}}

\usepackage{setspace}\setdisplayskipstretch{}

\usepackage{tikz}
\usetikzlibrary{decorations.pathmorphing,calc}
\usetikzlibrary{positioning}
\usepackage{wrapfig}
\begin{document}
\title{Connectivity-Preserving Minimum Separator in AT-free Graphs}
%
%
\author{Batya Kenig\institute{Technion, Israel Institute of Technology}}

\eat{
\author{First Author\inst{1}\orcidID{0000-1111-2222-3333} \and
Second Author\inst{2,3}\orcidID{1111-2222-3333-4444} \and
Third Author\inst{3}\orcidID{2222--3333-4444-5555}}
}
\authorrunning{Batya Kenig}
%
\institute{Technion, Israel Institute of Technology \email{batyak@technion.ac.il}}
\eat{
\institute{Princeton University, Princeton NJ 08544, USA \and
Springer Heidelberg, Tiergartenstr. 17, 69121 Heidelberg, Germany
\email{lncs@springer.com}\\
\url{http://www.springer.com/gp/computer-science/lncs} \and
ABC Institute, Rupert-Karls-University Heidelberg, Heidelberg, Germany\\
\email{\{abc,lncs\}@uni-heidelberg.de}}
}
\maketitle              
\begin{abstract}
	Let $A$ and $B$ be disjoint, non-adjacent vertex-sets in an undirected, connected graph $G$, whose vertices are associated with positive weights. 
We address the problem of identifying a minimum-weight subset of vertices $S\subseteq \nodes(G)$ that, when removed, disconnects $A$ from $B$ while preserving the internal connectivity of both $A$ and $B$.
We call such a subset of vertices a \e{connectivity-preserving}, or \e{safe} minimum $A,B$-separator. Deciding whether a safe $A,B$-separator exists is NP-hard by reduction from the \textsc{2-disjoint connected subgraphs} problem~\cite{van_t_hof_partitioning_2009}, and remains NP-hard even for restricted graph classes that include planar graphs~\cite{gray_removing_2012}, and $P_\ell$-free  graphs if $\ell\geq 5$~\cite{van_t_hof_partitioning_2009}.
In this work, we show that if $G$ is AT-free then in polynomial time we can find a safe $A,B$-separator of minimum weight, or establish that no safe $A,B$-separator exists.

\end{abstract}
\section{Introduction}
\label{sec:introduction}
Let $G$ be a simple, undirected, connected graph, whose vertices are associated with positive weights, and let $A,B\subseteq \nodes(G)$ be disjoint, non-adjacent subsets of $\nodes(G)$. That is, $G$ does not contain an edge that has an endpoint in $A$ and an endpoint in $B$. 
An $A,B$-separator is a subset of vertices $S\subseteq \nodes(G){\setminus} (A\cup B),$ such that in the graph $G\sminus S$ that results from $G$ by removing $S$ and its adjacent edges, $A$ and $B$ are disconnected; for every pair $a\in A$ and $b\in B$ there is no path in $G\sminus S$ between $a$ and $b$. We say that $S$ is a \e{minimal} $A,B$-separator if no strict subset of $S$ is an $A,B$-separator. 
The objective is to compute an $A,B$-separator with minimum weight to disconnect $A$ from $B$, while preserving the connectivity of $A$ and $B$ in $G\sminus S$. That is,
the graph $G\sminus S$ has two distinct, connected components $A_1$ and $B_1$ such that $A\subseteq A_1$ and $B\subseteq B_1$. For brevity, we call such a connectivity preserving $A,B$-separator a \e{safe} $A,B$-separator. 
In this paper, we study the problem of finding a safe $A,B$-separator of minimum weight, if one exists, or determine that no safe $A,B$-separator exists. We refer to this problem as \textsc{Min Safe Separator}.

Finding minimum separators, and minimum (edge) separators, under connectivity constraints is crucial in various domains, and has been studied extensively~\cite{DBLP:journals/corr/DuanJAX14,2b4537a3-d86e-332e-8839-fbf48c650a0f,DBLP:journals/networks/CullenbineWN13,DBLP:journals/jcss/DuanX14,DBLP:conf/icalp/BentertDFGK24}. In network security, for instance, a common challenge during denial-of-service attacks is to isolate compromised nodes from the rest of the network while preserving the connectivity among un-compromised nodes~\cite{DBLP:journals/corr/DuanJAX14,DBLP:journals/networks/CullenbineWN13}.
Connectivity preserving minimum separators are also used to model clustering problems with constraints enforcing clustering of certain objects together~\cite{DBLP:journals/jcss/DuanX14}. 
The connectivity preserving minimum vertex (edge) separator problem is defined as follows~\cite{DBLP:journals/jcss/DuanX14}. 
The input is a connected graph $G$ with positive vertex (or edge) weights, and three vertices $s,t,v$. The objective is to compute a vertex (edge) separator of minimum weight to disconnect $s$ from $t$, while preserving the connectivity of $s$ and $v$~\cite{DBLP:journals/jcss/DuanX14}. There are two variants of this problem, one that seeks a connectivity preserving minimum edge separator, and one that seeks a connectivity preserving minimum vertex separator. The latter corresponds to \textsc{min safe separator} where $|A|=2$ and $|B|=1$. It was shown in~\cite{DBLP:journals/jcss/DuanX14}, that even if $|A|=2$ and $|B|=1$, it is NP-complete to approximate the minimum safe vertex separator within $\alpha \log(|\nodes(G)|)$, for any constant $\alpha$. Interestingly, Duan and Xu~\cite{DBLP:journals/jcss/DuanX14} showed that the connectivity preserving minimum edge separator can be solved in polynomial time in planar graphs. Recently, Bentert et al.~\cite{DBLP:conf/icalp/BentertDFGK24} extended this result and presented a randomized algorithm that finds a connectivity preserving minimum-cardinality $A,B$-edge separator in planar graphs, in time $2^{|A|+|B|}\cdot |\nodes(G)|$.
\eat{
\textsc{Min Safe Separator} is closely related to the \textsc{Network Diversion} problem, which has been extensively studied by the networks and operations research communities~\cite{2b4537a3-d86e-332e-8839-fbf48c650a0f,DBLP:journals/networks/CullenbineWN13,DBLP:journals/jcss/DuanX14,DBLP:journals/corr/DuanJAX14}. The input to this problem is an undirected graph $G$, two terminal vertices $s,t\in \nodes(G)$, a distinguished vertex $v$, and a weight threshold $W$. The goal is to decide whether there exists a minimal $s,t$-separator, whose weight is at most $W$, that contains $v$~\cite{DBLP:journals/networks/CullenbineWN13}. Observe that a polynomial-time algorithm for the special case of $\textsc{min safe separator}$ where $|A|=|B|=2$, implies a polynomial time algorithm for \textsc{network diversion}. Indeed, $v$ belongs to a minimal $s,t$-separator $S$ if and only if there exist two non-adjacent neighbors of $v$, denoted $a$ and $b$, so that $S$ is a safe separator where $A=\set{s,a}$, and $B=\set{t,b}$. Iterating over the non-adjacent pairs of neighbors of $v$, and computing the appropriate minimum safe separator provides us with the minimum safe separator that contains $v$.
}

\textsc{Min Safe Separator} is the natural optimization variant of the intensively studied \textsc{2-disjoint connected subgraphs} problem~\cite{cygan_solving_2014,gray_removing_2012,paulusma_partitioning_2011,van_t_hof_partitioning_2009,DBLP:conf/wg/TelleV13,kern_disjoint_2022,DBLP:journals/tcs/GolovachKP13}.
The  \textsc{2-disjoint connected subgraphs} problem receives as input an undirected graph $G$, together with two disjoint subsets of vertices $A,B\subseteq \nodes(G)$. The goal is to decide whether there exist two disjoint subsets $A_1,B_1\subseteq \nodes(G)$, such that $A\subseteq A_1$, $B\subseteq B_1$ and $A_1$ and $B_1$ are connected in $G$.
The \textsc{2-disjoint connected subgraphs} problem remains NP-complete even in very restricted settings where $|A|=2$, $|B|$ is unbounded, and the input graph is a line graph (a subclass of claw-free graphs)~\cite{kern_disjoint_2022}. It is also NP-complete for planar graphs~\cite{gray_removing_2012}, and in many other settings~\cite{van_t_hof_partitioning_2009,kern_disjoint_2022}.
In general graphs, \textsc{2-Disjoint Connected Subgraphs} is NP-Complete even if $|A|=|B|=2$ by reduction from the \textsc{Induced Disjoint Paths Problem}~\cite{DBLP:conf/ipco/KawarabayashiK08}. 
It is easy to show that 
deciding whether a safe $A,B$-separator even exists is NP-complete by reduction from \textsc{2-Disjoint Connected Subgraphs}, and the formal proof is deferred to Section~\ref{sec:hardnessOfMinSafeSep} of the Appendix.

We consider \textsc{Min Safe Separator} in the class of \e{asteroidal triple-free graphs}, also known as \e{AT-free} graphs. An \e{asteroidal triple} is a set of three mutually non-adjacent vertices, such that every pair of vertices from this triple is joined by a path that avoids the neighborhood of the third. AT-free graphs are exactly those graphs that contain no such triple. AT-free graphs are intensively studied and include as a subclass the set of cobipartite graphs, cocomparability graphs, cographs, interval graphs, permutation graphs, and trapezoid graphs~\cite{DBLP:journals/siamdm/CorneilOS97,DBLP:journals/tcs/GolovachKP13}. In previous work, Golovach, Kratsch, and Paulusma~\cite{DBLP:journals/tcs/GolovachKP13} presented a dynamic programming algorithm for \textsc{$k$-Disjoint Connected Subgraphs} in AT-free graphs, where $k$ is fixed. For the case of $k=2$ their algorithm involves ``guessing'' at least 12 vertices and has a runtime of $O(n^{15})$ for an $n$-vertex, AT-free graph  (cf.~\cite{DBLP:journals/tcs/GolovachKP13}). Our algorithm goes well beyond deciding whether a safe $A,B$-separator exists, and actually finds a minimum, safe $A,B$-separator (or decides that none exist) orders of magnitude faster.
Towards this goal, we prove new properties of minimal separators in AT-free graphs that may be of independent interest, and show how questions regarding the existence of disjoint connected subgraphs can be translated to questions regarding the existence of certain minimal separators in the graph.
\eat{
Our techniques also differ from those of~\cite{DBLP:journals/tcs/GolovachKP13}, and use the properties of minimal separators in AT-free graphs.}
\begin{theorem}
	Let $G$ be a simple, undirected, connected, weighted AT-free graph, and let $A,B\subseteq \nodes(G)$ be a pair of non-empty, disjoint, vertex-sets. There is an algorithm that finds a safe $A,B$-separator of minimum weight, or establishes that no safe $A,B$-separator exists in time $O(n^4 \cdot T(n,m))$ where $n=|\nodes(G)|,m=|\edges(G)|$, and $T(n,m)$ is the time to find a minimum $s,t$-separator in $G$ for some pair of vertices $s,t \in \nodes(G)$.
\end{theorem}
The problem of finding a minimum $s,t$-separator in an undirected graph can be reduced, by standard techniques~\cite{DBLP:books/daglib/0032640}, to the problem of finding a minimum $s,t$-cut, or maximum flow from $s$ to $t$. 
Following a sequence of improvements to max-flow algorithms in the past few years~\cite{DBLP:conf/stoc/LiuS20,DBLP:conf/focs/KathuriaLS20,DBLP:conf/stoc/BrandLLSS0W21}, the current best running time is $O(m^{1+o(1)})$~\cite{Chen2022,DBLP:conf/focs/Brand0PKLGSS23}. 
\newline

\textbf{Organization.} The rest of this paper is organized as follows. Following preliminaries in Section~\ref{sec:Preliminaries}, we establish some basic results on minimal $s,t$-separators, and on minimal separators between vertex sets, in Section~\ref{sec:minimalSeparators}. In Section~\ref{sec:AlgOverview},
we give an overview of the algorithm, high-level pseudo-code, and map the results that need to be proved to establish its correctness and runtime guarantee.
In Section~\ref{sec:charMinlSeps} we prove several results about minimal $s,t$-separators and minimal $s,t$-separators in AT-free graphs in particular. The main theorem behind the algorithm is proved in Section~\ref{sec:MainThm}, where we also present the pseudo-code of the main component. Due to space restrictions, some of the proofs and technical details are deferred to the Appendix.

\section{Preliminaries and Notation}
\label{sec:Preliminaries}
Let $G$ be an undirected graph with nodes $\nodes(G)$ and edges $\edges(G)$, where $n=|\nodes(G)|$, and $m=|\edges(G)|$. We assume a positive weight function on the vertices $w: \nodes(G) \rightarrow \mathbb{Z}^+$. 
We also assume, without loss of generality, that $G$ is connected. 
For $A,B\subseteq \nodes(G)$, we abbreviate $AB\eqdef A\cup B$; for $v\in \nodes(G)$ we abbreviate $vA\eqdef \set{v}\cup A$. Let $v\in V$. We denote by $N_G(v)\eqdef\set{u\in\nodes(G) : (u,v)\in \edges(G)}$ the neighborhood of $v$, and by $N_G[v]\eqdef N_G(v)\cup \set{v}$ the \e{closed} neighborhood of $v$.
For a subset of vertices $T\subseteq \nodes(G)$, we denote by $N_G(T)\eqdef \bigcup_{v\in T}N_G(v){\setminus}T$, and $N_G[T]\eqdef N_G(T)\cup T$.
We denote by $G[T]$ the subgraph of $G$ induced by $T$. Formally, $\nodes(G[T])=T$, and $\edges(G[T])=\set{(u,v)\in \edges(G): \set{u,v}\subseteq T}$. For a subset $S\subseteq \nodes(G)$, we abbreviate $G\sminus S\eqdef G[\nodes(G){\setminus} S]$; for $v\in \nodes(G)$, we abbreviate $G\sminus v\eqdef G\sminus \set{v}$. \eat{For a pair of vertices $x,y\in \nodes(G)$, we denote by $G\splus(x,y)$ the graph that results from $G$ by adding the edge $(x,y)$. For an edge $(x,y)\in \edges(G)$, we denote by $G\sminus (x,y)$ the graph that results from $G$ by removing the edge $(x,y)$.}
We say that $G'$ is a \e{subgraph} of $G$ if it results from $G$ by removing vertices and edges; formally, $\nodes(G')\subseteq \nodes(G)$ and $\edges(G')\subseteq \edges(G)$. In that case, we also say that $G$ is a \e{supergraph} of $G'$.
\eat{
Let $T\subseteq \nodes(G)$, and $t\in \nodes(G)$. 
By \e{merging} $T$ into  vertex $t\in \nodes(G)$, we refer to the operation that adds an edge between $t$ and every vertex in $N_G[T]$. Formally, merging $T$ to $t$ results in the graph $G'$ where: 
\begin{align}
	\nodes(G')\eqdef \nodes(G) && \edges(G')\eqdef \edges(G)\cup \set{(t,u):u\in N_G[T]} \label{eq:mergeDef}
\end{align}
We note the distinction from \e{vertex contraction} or \e{vertex identification}\footnote{\href{https://mathworld.wolfram.com/VertexContraction.html}{https://mathworld.wolfram.com/VertexContraction.html}} where the vertices of $T$ are replaced by a single vertex $t$ that is made adjacent to $N_G(T)$.}
\eat{
Let $u\in \nodes(G)$. We denote by $\sat(G,u)$ the graph that results from $G$ by adding all edges between vertices in $N_G(u)$. Formally:
\begin{align*}
	\label{eq:SAT}
	\nodes(\sat(G,u))=\nodes(G) &&\mbox{ and }&& \edges(\sat(G,u))=\edges(G) \cup\set{(x,y): x,y \in N_G(u)}
\end{align*}
}

Let $(u,v)\in \edges(G)$. \e{Contracting the edge $(u,v)$ to vertex $u$} results in a new graph $G'$ where:
\begin{align*}
	\nodes(G')=\nodes(G){\setminus}\set{v} && \mbox{ and } && \edges(G')=\edges(G\sminus v)\cup \set{(u,x):x\in N_G(v)}
\end{align*}
Let $u,v \in \nodes(G)$. A \e{simple path} between $u$ and $v$, called a $u,v$-path, is a finite sequence of distinct vertices $u=v_1,\dots,v_k=v$ where, for all $i\in [1,k-1]$, $(v_{i},v_{i+1})\in \edges(G)$, and whose ends are $u$ and $v$. A $u,v$-path is \e{chordless} or \e{induced} if $(v_i,v_j)\notin \edges(G)$ whenever $|i-j|>1$. 

A subset of vertices $A\subseteq \nodes(G)$ is said to be \e{connected} in $G$ if $G[A]$ contains a  path between every pair of vertices in $A$. A subset of vertices $A\subseteq \nodes(G)$ is said to be a \e{connected component} of $G$ if $A$ is connected, and $A'$ is not connected for every subset of vertices $A\subset A'\subseteq \nodes(G)$ that properly contains $A$.
Let $A\subseteq \nodes(G)$ and $u\in \nodes(G){\setminus}A$, where $G[uA]$ is connected. \e{Contracting the connected vertex-set $uA$ to vertex $u$} results in a new graph $G'$ where $\nodes(G')=\nodes(G){\setminus}A$ and $ \edges(G')=\edges(G\sminus A)\cup \set{(u,a):a\in N_G(A)}$. It is easy to see that contracting a connected vertex-set $uA$ to $u$ is equivalent to multiple edge contractions. 

 Let $V_1,V_2\subseteq \nodes(G)$ denote two disjoint vertex subsets of $\nodes(G)$. We say that $V_1$ and $V_2$ are adjacent if there is at least one pair of adjacent vertices $v_1 \in V_1$ and $v_2\in V_2$. We say that there is a path between $V_1$ and $V_2$ if there exist vertices $v_1 \in V_1$ and $v_2\in V_2$ such that there is a path between $v_1$ and $v_2$. 

 Three mutually non-adjacent vertices of a graph form an \e{asteroidal triple} if every two of them are connected by a path avoiding the neighborhood of the third. A graph is \e{AT-free} if it does not contain any asteroidal triple. By this definition, if $G$ is AT-free, then every induced subgraph of $G$ is AT-free.

\section{Minimal Separators} 
\label{sec:minimalSeparators}
Let $s,t \in 
\nodes(G)$. 
For $X \subseteq \nodes(G)$, we let $\cc(G\sminus X)$ denote the set of connected components of $G\sminus X$. The vertex set $X$ is called a \e{separator} of $G$ if $|\cc(G\sminus X)|\geq 2$, an \e{$s,t$-separator} if $s$ and $t$ are in different connected components of $\cc(G\sminus X)$, and a \e{minimal $s,t$-separator} if no proper subset of $X$ is an $s,t$-separator of $G$. For an $s,t$-separator $X$, we denote by $C_s(G\sminus X)$ and $C_t(G\sminus X)$ the connected components of $\cc(G\sminus X)$ containing $s$ and $t$ respectively.
In other words, $C_s(G\sminus X)=\set{v\in \nodes(G): \text{ there is a path from }s\text{ to }v\text{ in }G\sminus X}$.

\begin{citedlemma}{\cite{DBLP:journals/ijfcs/BerryBC00}}\label{lem:fullComponents}
	An $s,t$-separator $X\subseteq \nodes(G)$ is a minimal $s,t$-separator if and only if $N_G(C_s(G\sminus X))=N_G(C_t(G\sminus X))=X$. \eat{, in which case $C_s(G\sminus X)$ and $C_t(G\sminus X)$ are called \e{full components} of $\cc(G\sminus X)$.}\eat{ 
	there are two connected components $C_s\eqdef C_s(G,X),C_t\eqdef C_t(G,X)\in \cc_G(X)$, such that $s\in C_s$, $t \in C_t$, and $N_G(C_s)=N_G(C_t)=X$; $C_s$ and $C_t$ are called \e{full components} of $\cc_G(X)$.}
\end{citedlemma}
A subset $X\subseteq \nodes(G)$ is a \e{minimal separator} of $G$ if there exist a pair of vertices $u,v \in \nodes(G)$ such that $X$ is a minimal $u,v$-separator. A connected component $C\in \cc(G\sminus X)$ is called a \e{full component} of $X$ if $N_G(C)=X$. By Lemma~\ref{lem:fullComponents}, $X$ is a minimal $u,v$-separator if and only if the components $C_u(G\sminus X)$ and $C_v(G\sminus X)$ are full.
We denote by $\minlsepst{G}$ the set of minimal $s,t$-separators of $G$, and by $\minlsep{}{G}$ the set of minimal separators of $G$.

\subsection{Minimal Separators Between Vertex-Sets}
\label{sec:minlsepsVertexSets}
Let $A,B \subseteq \nodes(G)$ be disjoint and non-adjacent. A subset $S\subseteq \nodes(G){\setminus}AB$ is an $A,B$-separator if, in the graph $G\sminus S$, there is no path between $A$ and $B$. We say that $S$ is a minimal $A,B$-separator if no proper subset of $S$ is an $A,B$-separator. We denote by $\minlsep{A,B}{G}$ the set of minimal $A,B$-separators of $G$.
In Section~\ref{sec:minsepsvertexsets} of the Appendix, we prove the following two technical lemmas that show how finding minimal separators between vertex-sets can be reduced to the problem of finding minimal separators between singleton vertices. 
\def\simpABlemma{
	Let $A$ and $B$ be two disjoint, non-adjacent subsets of $\nodes(G)$. Then $S\in \minlsep{A,B}{G}$ if and only if $S$ is an $A,B$-separator, and for every $w\in S$, there exist two connected components $C_A,C_B\in \cc(G\sminus S)$ such that $C_A\cap A\neq \emptyset$, $C_B\cap B\neq \emptyset$, and $w\in N_G(C_A)\cap N_G(C_B)$.
}
\begin{lemma}
	\label{lem:simpAB}
	\simpABlemma
\end{lemma}
Observe that Lemma~\ref{lem:simpAB} implies Lemma~\ref{lem:fullComponents}. By Lemma~\ref{lem:simpAB}, it holds that $S\in \minlsepst{G}$ if and only if $S$ is an $s,t$-separator and $S\subseteq N_G(C_s(G\sminus S))\cap N_G(C_t(G\sminus S))$. By definition, $N_G(C_s(G\sminus S))\subseteq S$ and $N_G(C_t(G\sminus S))\subseteq S$, and hence $S=N_G(C_s(G\sminus S))\cap N_G(C_t(G\sminus S))$, and $S=N_G(C_s(G\sminus S))=N_G(C_t(G\sminus S))$.
\def\lemMinlsASep{
	Let $A\subseteq \nodes(G){\setminus}\set{s,t}$. Let $H$ be the graph that results from $G$ by adding all edges between $s$ and $N_G[A]$. That is, $\edges(H)=\edges(G)\cup \set{(s,v):v\in N_G[A]}$. Then $\minlsep{sA,t}{G}=\minlsepst{H}$.
}
\begin{lemma}
	\label{lem:MinlsASep}
	\lemMinlsASep
\end{lemma}

\def\contractEdgesMinlSep{
	Let $u\in N_G(s)$, and let $H$ be the graph that results from $G$ by contracting the edge $(s,u)$ to vertex $s$. Then $\minlsepst{H}=\set{S\in \minlsepst{G}: u\in C_s(G\sminus S)}$.
}

\subsection{Minimal $s,t$-Separators: Some Basic Properties}
\label{sec:MaintainingMinlseps}
The following are basic results used by our algorithms. Due to space restrictions, the proofs of Lemmas~\ref{lem:contract} and~\ref{lem:inclusionCsCt} are deferred to Section~\ref{sec:proofsFromPrelims} of the Appendix.
\def\thmSAT{
		Let $u\in \nodes(G)$. Then: (1) $\minlsepst{\sat(G,u)}=\set{S\in \minlsepst{G}: u\notin S}$, and (2) If $G$ is AT-free then $\sat(G,u)$ is AT-free.
}
\def\thmContract{
		Let $u\in N_G(s)$, and let $H$ be the graph that results from $G$ by contracting the edge $(s,u)$ to vertex $s$. Then: (1) $\minlsepst{H}=\set{S\in \minlsepst{G}: u\in C_s(G\sminus S)}$, and (2) If $G$ is AT-free then $H$ is AT-free.
}
\def\lemContract{
		Let $s,t\in \nodes(G)$, and $A\subseteq \nodes(G){\setminus}\set{s,t}$ such that $G[sA]$ is connected. Let $H$ be the graph where $\nodes(H)=\nodes(G){\setminus}A$ that results from $G$ by contracting all edges in $G[sA]$. Then (1) $\minlsepst{H}=\set{S\in \minlsepst{G}: A\subseteq  C_s(G\sminus S)}$, and (2) If $S\in \minlsepst{H}$, then $C_s(G\sminus S)=C_s(H\sminus S)\cup A$ and $C_t(G\sminus S)=C_t(H\sminus S)$. \eat{Also, if $G$ is AT-free then $H$ is AT-free.}
}
\begin{lemma}
	\label{lem:contract}
	\lemContract
\end{lemma} 

\def\inclusionCsCt{
	Let $s,t\in \nodes(G)$, and let $S,T\in \minlsepst{G}$. The following holds:
	\begin{align*}
		C_s(G\sminus S)\subseteq C_s(G\sminus T) &&\Longleftrightarrow&& S\subseteq T\cup C_s(G\sminus T) &&\Longleftrightarrow&&   T\subseteq S\cup C_t(G\sminus S). 
	\end{align*}
}
\begin{lemma}
	\label{lem:inclusionCsCt}
	\inclusionCsCt
\end{lemma}

\def\closetosminlsep{
Let $G$ be a claw-free graph, and let $S\in \minlsepst{G}$, where $S\subseteq N_G(s)$. Let $u\in S$, and let $G'$ denote the graph that results from $G$ by subdividing the edge $(s,u)$ to $(s,u')$ and $(u',u)$. Then:
\begin{enumerate}
	\item $\sat(G',s)$ is claw-free.
	\item $\minlsepst{\sat(G',s)}=\minlsepst{G}\cup \set{T{\setminus}\set{u}\cup \set{u'}: T\in \minlsepst{G}, u\in T}$
\end{enumerate}
}

\section{Algorithm Overview}
\label{sec:AlgOverview}
\def\F{\mathcal{F}}
In this Section, we give an overview of the algorithm, and map the results that need to be proved to establish its correctness and runtime guarantee.
\begin{definition}
	\label{def:closeToA}
	Let $A\subseteq \nodes(G)$. We say that  $S\in \minlsepst{G}$ is \e{close to $sA$} if:
	\begin{enumerate}
		\item $A\subseteq C_s(G\sminus S)$.
		\item For every $T\in \minlsepst{G}{\setminus}\set{S}$, if $A\subseteq C_s(G\sminus T)$ then $C_s(G\sminus T) \not\subseteq C_s(G\sminus S)$.
	\end{enumerate}
\end{definition}
We denote by $\F_{sA}(G)$ the minimal $s,t$-separators that are close to $sA$. By Definition~\ref{def:closeToA}, we have that if $S\in \minlsepst{G}$ where $A\subseteq C_s(G\sminus S)$, then there exists a $T\in \F_{sA}(G)$, such that $A\subseteq C_s(G\sminus T)\subseteq C_s(G\sminus S)$. We make this formal in Lemma~\ref{lem:FsaExists} whose proof is deferrred to Section~\ref{sec:WGAppendixAlgOverview} of the Appendix.
\def\FsaExists{	Let $A\subseteq \nodes(G)$, and let $S\in \minlsepst{G}$ where $A\subseteq C_s(G\sminus S)$. There exists a $T\in \F_{sA}(G)$ where $C_s(G\sminus T)\subseteq C_S(G\sminus S)$.}
\begin{lemma}
	\label{lem:FsaExists}
	\FsaExists
\end{lemma}
\eat{
\begin{proof}
	By induction on $|C_s(G\sminus S)|$. If $|C_s(G\sminus S)|=|sA|$, then $C_s(G\sminus S)=sA$. By definition, $S\in \F_{sA}(G)$. Suppose the claim holds for the case where $|C_s(G\sminus S)|\leq k$ for some $k\geq |sA|$, we prove for the case where $|C_s(G\sminus S)|= k+1$. If $S\in \F_{sA}(G)$, then we are done. Otherwise, there exists a $S'\in \minlsepst{G}{\setminus}\set{S}$ where $A\subseteq C_s(G\sminus S')$ and $C_s(G\sminus S')\subseteq C_s(G\sminus S)$. Since $S'\neq S$, then $C_s(G\sminus S')\subset C_s(G\sminus S)$. Since $|C_s(G\sminus S')|< |C_s(G\sminus S)|=k+1$, then by the induction hypothesis, there exists a $T\in \F_{sA}(G)$ where $C_s(G\sminus T)\subseteq C_s(G\sminus S')\subset C_s(G\sminus S)$. \qed
\end{proof}}
\renewcommand{\algorithmicrequire}{\textbf{Input: }}
\renewcommand{\algorithmicensure}{\textbf{Output: }}
\eat{
\begin{algserieswide}
	{H}{Algorithm for returning a minimum, safe $A,B$-separator if one exists and $\bot$ otherwise. \label{fig:MinSafeSep}}
	\begin{insidealgwide}{MinSafeSep}{$G$, $A$, $B$}
		\REQUIRE{AT-free graph $G$, $\emptyset \subset A,B \subseteq \nodes(G)$.}
		\ENSURE{A minimum, safe $A,B$-separator or $\bot$ if none exist}.
		\IF{$A\cap N_G[B] \neq \emptyset$}
			\RETURN $\bot$ \label{line:retBot}
		\ENDIF
		\STATE $G\gets G\sminus (N_G(A)\cap N_G(B))$ \label{line:alwaysIncluded}
		\STATE Let $s\in A$, and $t\in B$
		\STATE $\F_{sA}(G) \gets \algname{CloseTo}(G,s,t,A{\setminus}\set{s})$ \label{line:FA}
		\STATE $\F_{tB}(G) \gets \algname{CloseTo}(G,t,s,B{\setminus}\set{t})$ \label{line:FB}
		\STATE $R \gets \bot$
		\FORALL{$S_A\in \F_{sA}(G)$ and $S_B\in \F_{tB}(G)$} \label{line:loopStarts}
			\IF{$S_A\subseteq S_B \cup C_s(G\sminus S_B)$}
				\STATE Let $G(S_A,S_B)$ be the graph that results from $G$ by contracting $C_s(G\sminus S_A)$ to $s$ and $C_t(G\sminus S_{B})$ to $t$.\label{line:GAB}\COMMENT{Lemma~\ref{lem:contract}} 
				\STATE $T_{AB}\gets \algname{MinSep}(G(S_A,S_B),s,t)$ \label{line:minSep}
				\IF{$R=\bot$ or $|R|>|T_{AB}|$}
					\STATE $R\gets T_{AB}$
				\ENDIF
			\ENDIF
		\ENDFOR \label{line:loopEnds}
		\RETURN $R\cup (N_G(A)\cap N_G(B))$ \label{line:return}
	\end{insidealgwide}
\end{algserieswide}
}
If no restrictions are made to $A\subseteq \nodes(G)$, there may be an unbounded number of minimal $s,t$-separators that are close to $sA$. If $G[sA]$ is connected then, following a result by Takata~\cite{DBLP:journals/dam/Takata10}, the minimal $s,t$-separator close to $sA$ is unique and can be found in time $O(m)$.
\begin{algorithm}[t]
	\SetAlgoLined
	\SetNoFillComment
	\KwIn{Connected, weighted, AT-free graph $G$, and $\emptyset \subset A,B\subseteq \nodes(G)$.}
	\KwOut{A minimum-weight, safe $A,B$-separator, or $\bot$ if none exist.}
	\SetKwProg{Algorithm2}{Algorithm1}{:}{}
	\lIf{$A\cap N_G[B]\neq \emptyset$}{
		\Return $\bot$}  \label{line:retBot}
	$G \gets G\sminus(N_G(A)\cap N_G(B))$ \label{line:alwaysIncluded} \\
		Let $s\in A$, and $t\in B$ \\
	$\F_{sA}(G) \gets \algname{CloseTo}(G,s,t,A{\setminus}\set{s})$ \label{line:FA}\\
	$\F_{tB}(G) \gets \algname{CloseTo}(G,t,s,B{\setminus}\set{t})$ \label{line:FB} \\
	Initialize $R\gets \bot$\\
	\ForAll{$S_A\in \F_{sA}(G)$ and $S_B\in \F_{tB}(G)$}{ \label{line:loopStarts}
		\If{$S_A\subseteq S_B \cup C_s(G\sminus S_B)$}{
			Let $G(S_A,S_B)$ be the graph that results from $G$ by contracting $C_s(G\sminus S_A)$ to $s$ and $C_t(G\sminus S_{B})$ to $t$ \tcp*[r]{Lemma~\ref{lem:contract}.} \label{line:GAB}
			$T_{AB}\gets \algname{MinSep}(G(S_A,S_B),s,t)$ \label{line:minSep} \\
			\lIf{$R=\bot$ or $w(R)>w(T_{AB})$}{
				$R\gets T_{AB}$}}}\label{line:forendMinSafe}
	\Return $R \cup (N_G(A)\cap N_G(B))$ \label{line:return}
	\caption{$\algname{MinSafeSep}$.\label{fig:MinSafeSep}}
\end{algorithm}
\begin{citedlemma}{\cite{DBLP:journals/dam/Takata10}}
	\label{lem:uniqueCloseVertexSet}
	Let $A\subseteq \nodes(G)$, where $G[sA]$ is connected. If $sA\cap N_G[t]=\emptyset$ then $N_G(sA)$ contains a unique minimal $s,t$-separator, which can be found in $O(m)$ time.
\end{citedlemma}
\begin{corollary}
	\label{corr:singleCloseSet}
	Let $A\subseteq \nodes(G)$ such that $G[sA]$ is connected. If $sA\cap N_G[t]=\emptyset$ there exists a unique minimal $s,t$-separator that is close to $sA$, which can be found in $O(m)$ time.
\end{corollary}
\begin{proof}
	Let $T\in \minlsepst{G}$ where $T\subseteq N_G(sA)$. By Lemma~\ref{lem:uniqueCloseVertexSet}, $T$ is unique. Let $S\in \minlsepst{G}$ that is close to $sA$. By Definition~\ref{def:closeToA}, $A\subseteq C_s(G\sminus S)$. Therefore, $T \subseteq N_G(sA)\subseteq S\cup C_s(G\sminus S)$. By Lemma~\ref{lem:inclusionCsCt}, we have that $C_s(G\sminus T)\subseteq C_s(G\sminus S)$. By Definition~\ref{def:closeToA}, $S=T$. \qed
\end{proof}
In Section~\ref{sec:MainThm}, we prove that if $G$ is AT-free then $|\F_{sA}(G)|\leq n^2$ ($|\F_{tB}(G)|\leq n^2$), and the set $\F_{sA}(G)$ ($\F_{tB}(G)$) can be computed in time $O(n^2m)$. 
\def\mainThmATFree{
	Let $G$ be an AT-free graph, $s,t\in \nodes(G)$ two distinguished vertices, and $A\subseteq \nodes(G){\setminus}\set{s,t}$. 
	Let $T_s\in \minlsepst{G}$ where $T_s \subseteq N_G(s)$.
	If $A\subseteq C_s(G\sminus T_s)\cup T_s \cup C_t(G\sminus T_s)$, then there are at most $n$ minimal $s,t$-separators that are close to $sA$, and they can be found in time $O(nm)$.
	Otherwise, at most $n^2$ minimal $s,t$-separators are close to $sA$, and they can be found in time $O(n^2m)$.
	
}
\begin{theorem}
	\label{thm:mainThmATFree}
	\mainThmATFree
\end{theorem}
In Section~\ref{sec:MainThm}, we prove Theorem~\ref{thm:mainThmATFree}, and present algorithm $\algname{CloseTo}$ that receives as input an AT-free graph $G$, vertices $s,t\in \nodes(G)$ and subset $A\subseteq \nodes(G)$ ($B\subseteq \nodes(G)$), and computes $\F_{sA}(G)$ ($\F_{tB}(G)$) in time $O(n^2m)$. If $T_s\in \minlsepst{G}$ is the unique minimal $s,t$-separator where $T_s\subseteq N_G(s)$ (Lemma~\ref{lem:uniqueCloseVertexSet}), and if $A\subseteq T_s\cup C_t(G\sminus T_s)$, then algorithm $\algname{CloseTo}$ computes $\F_{sA}(G)$ (or $\F_{tB}(G)$) in time $O(nm)$. 

We now describe the algorithm for \textsc{Min safe sep}, presented in Algorithm~\ref{fig:MinSafeSep}, that receives as input a vertex-weighted, AT-free graph $G$, and a pair of vertex sets $A,B\subseteq \nodes(G)$. If $A\cap N_G[B]\neq \emptyset$, then no $A,B$-separator exists and the algorithm returns $\bot$ in line~\ref{line:retBot}. Every $A,B$-separator must include $N_G(A)\cap N_G(B)$. Therefore, the algorithm processes the graph $G\sminus (N_G(A)\cap N_G(B))$ (line~\ref{line:alwaysIncluded}). Since $G$ is AT-free, then $G\sminus (N_G(A)\cap N_G(B))$ is also AT-free.
The algorithm relates minimal $s,t$-separators to minimal, safe $A,B$-separators using the following.
\def\lemminSafeSepOverview1{
	A subset $S\subseteq \nodes(G)$ is a safe, minimal $A,B$-separator if and only if for every pair of vertices $s\in A$ and $t\in B$ it holds that $S\in \minlsepst{G}$ where $A\subseteq C_s(G\sminus S)$ and $B\subseteq C_t(G\sminus S)$.
}
\begin{lemma}
	\label{lem:minSafeSepOverview1}
	\lemminSafeSepOverview1
\end{lemma}
\eat{
\begin{proof}
	Let $s\in A$, $t\in B$.	If $S\in \minlsepst{G}$ where $A\subseteq C_s(G\sminus S)$ and $B\subseteq C_t(G\sminus S)$ then clearly $S$ is a safe $A,B$-separator. 
	By Lemma~\ref{lem:fullComponents}, it holds that $S=N_G(C_s(G\sminus S))\cap N_G(C_t(G\sminus S))$.
	By Lemma~\ref{lem:simpAB}, it holds that $S$ is a minimal, safe $A,B$-separator.
	
	Now, let $S$ be a minimal, safe $A,B$-separator, where $C_A,C_B\in \cc(G\sminus S)$ contain $A$ and $B$ respectively. 
	By Lemma~\ref{lem:simpAB}, it holds that $S=N_G(C_A)\cap N_G(C_B)$. 
	By Lemma~\ref{lem:fullComponents}, $S\in \minlsepst{G}$ for every pair of vertices $s\in A$ and $t\in B$, where $C_s(G\sminus S)=C_A$ and $C_t(G\sminus S)=C_B$.\qed
\end{proof}
}
Take any $S\in \minlsepst{G}$ such that $A\subseteq C_s(G\sminus S)$ and $B\subseteq C_t(G\sminus S)$. By Lemma~\ref{lem:FsaExists}, there exists an $S_A\in \F_{sA}(G)$ such that $C_s(G\sminus S_A)\subseteq C_s(G\sminus S)$, and an $S_B\in \F_{tB}(G)$ such that $C_t(G\sminus S_B)\subseteq C_t(G\sminus S)$. Let $G(S_A,S_B)$ denote the graph that results from $G$ by contracting $C_s(G\sminus S_A)$ to vertex $s$ and $C_t(G\sminus S_B)$ to vertex $t$. By Lemma~\ref{lem:minSafeSepOverview1}, and Lemma~\ref{lem:contract}, every $T\in \minlsepst{G(S_A,S_B)}$ is a safe, minimal $A,B$-separator.
Consequently, by Lemma~\ref{lem:FsaExists} and Lemma~\ref{lem:minSafeSepOverview1}, we have that $S$ is a minimal, safe, $A,B$-separator if and only if $S\in \minlsepst{G(S_A,S_B)}$ for some pair of minimal $s,t$-separators $S_A\in \F_{sA}(G)$ and $S_B\in \F_{tB}(G)$. Moreover, $S$ is a minimum, safe $A,B$-separator if and only if $S$ is a minimum $s,t$-separator of $G(S_A,S_B)$ for some pair of minimal $s,t$-separators $S_A\in \F_{sA}(G)$ and $S_B\in \F_{tB}(G)$. 
The loop in lines~\ref{line:loopStarts}-\ref{line:forendMinSafe} runs over all pairs $S_A\in \F_{sA}(G)$ and $S_B\in \F_{tB}(G)$, generates the graph $G(S_A,S_B)$ in line~\ref{line:GAB}, and
finds a minimum-weight $s,t$-separator of $G(S_A,S_B)$ in line~\ref{line:minSep}. The algorithm returns the minimum over all combinations of $S_A\in \F_{sA}(G)$ and $S_B\in \F_{tB}(G)$ in line~\ref{line:return}.
\begin{theorem}
	\label{thm:minSafe}
	Given a weighted, connected, AT-free graph $G$, and two vertex-sets $A,B\subseteq \nodes(G)$, Algorithm $\algname{MinSafeSep}$ returns a minimum-weight, safe $A,B$-separator if one exists, or $\bot$ otherwise, and runs in time $O(|\F_{sA}(G)|\cdot |\F_{tB}(G)|\cdot T(n,m))$, where $s\in A,t\in B$, and $T(n,m)$ is the time to compute a minimum-weight $s,t$-separator.
\end{theorem}
Theorem~\ref{thm:mainThmATFree} establishes that in AT-free graphs, if no assumptions are made to the input vertex-sets $A,B\subseteq \nodes(G)$, then $|\F_{sA}(G)|\leq n^2$, $|\F_{tB}(G)|\leq n^2$, and that $\F_{sA}(G)$ and $\F_{tB}(G)$ can be computed in time $O(n^2m)$. It further establishes that if there exist vertices $s\in A$, and $t\in B$ such that $A{\setminus}\set{s}\subseteq C_s(G\sminus T_s)\cup T_s \cup C_t(G\sminus T_s)$ and $B{\setminus}\set{t}\subseteq C_t(G\sminus T_t)\cup T_t\cup C_s(G\sminus T_t)$ where $T_s,T_t\in \minlsepst{G}$ are the unique minimal $s,t$-separators close to $s$ and $t$, respectively (Lemma~\ref{lem:uniqueCloseVertexSet}), then $|\F_{sA}(G)|\leq n$, $|\F_{tB}(G)|\leq n$, and $\F_{sA}(G)$ and $\F_{tB}(G)$ can be computed in time $O(nm)$. It immediately follows from Theorem~\ref{thm:mainThmATFree} and Theorem~\ref{thm:minSafe}:
\begin{corollary}
	\label{corr:minSafe}
	Let $G$ be a weighted, connected, AT-free graph, and $A,B\subseteq \nodes(G)$. If there exist vertices $s\in A$ and $t\in B$ such that $A{\setminus}\set{s}\subseteq C_s(G\sminus T_s)\cup T_s \cup C_t(G\sminus T_s)$ and $B{\setminus}\set{t}\subseteq C_t(G\sminus T_t)\cup T_t\cup C_s(G\sminus T_t)$ then
	Algorithm $\algname{MinSafeSep}$ returns a minimum-weight, safe $A,B$-separator if one exists, or $\bot$ otherwise, in time $O(n^2T(n,m))$. Otherwise, the runtime is $O(n^4T(n,m))$. 
\end{corollary}
\eat{From Theorem~\ref{thm:minSafe} it follows that the runtime of $\algname{MinSafeSep}$ is in $O(n^4\cdot T(n,m))$ where} Recall that  $T(n,m)=O(m^{1+o(1)})$ is the time to compute a minimum-weight $s,t$-separator~\cite{Chen2022}. The rest of this paper is devoted to proving Theorem~\ref{thm:mainThmATFree}.
\eat{
\def\mainThmATFree{
	Let $G$ be an AT-free graph, $s,t\in \nodes(G)$ two distinguished vertices, and $A\subseteq \nodes(G){\setminus}\set{s,t}$. There exist at most $n^2$ minimal $s,t$-separators that are close to $sA$, and they can be found in time $O(n^2m)$.
}
\begin{theorem}
	\label{thm:mainThmATFree}
	\mainThmATFree
\end{theorem}
}
\eat{
The rest of this paper is devoted to proving that for an AT-free graph $G$, and a subset of vertices $sA \subseteq \nodes(G)$, the set $\F_{sA}(G)$ can be computed in time $O(n^2m)$ and $|\F_{sA}(G)|\leq n^2$. }
\section{Essential Findings: Minimal $s,t$-separators}
\label{sec:charMinlSeps}
In this Section, we prove several results about minimal $s,t$-separators that are crucial for proving Theorem~\ref{thm:mainThmATFree}.\eat{, which establishes that for an AT-free graph $G$, and a subset of vertices $A \subseteq \nodes(G)$, the set $\F_{sA}(G)$ can be computed in time $O(n^2m)$ and $|\F_{sA}(G)|\leq n^2$.} 
In Section~\ref{sec:moreResults}, we establish a result concerning minimal $s,t$-separators close to $sA$, where $A\subseteq \nodes(G)$ (Definition~\ref{def:closeToA}).
In Section~\ref{sec:MinlSepsATFree}, we establish results on minimal $s,t$-separators in AT-free graphs. 
Some of the proofs are deferred to Section~\ref{sec:AppendixProofsEssentialMiinlstseps} of the Appendix.

\def\crossingCharacterizinglem{
	Let $S,T\in \minlsep{}{G}$. Then $S\sharp T$ if and only if $T\cap C\neq \emptyset$ for every full connected component of $G\sminus S$.
}

\subsection{Results on Close Minimal $s,t$-separators}
\label{sec:moreResults}

\def\technicalLemmaCloseToA{
	Let $T_s\in \minlsepst{G}$ where $T_s\subseteq N_G(s)$. Let $A\subseteq \nodes(G){\setminus}(T_s \cup C_s(G\sminus T_s) \cup N_G[t])$ such that $T_s \subseteq N_G(a)$ for every $a\in A$. There are at most $|T_s|$ minimal $s,t$-separators that are close to $sA$, which can be found in time $O(|T_s|\cdot m)$.
}
\begin{lemma}
	\label{lem:technicalLemmaCloseToA}
	\technicalLemmaCloseToA
\end{lemma}
To illustrate Lemma~\ref{lem:technicalLemmaCloseToA}, consider Figure~\ref{fig:illustrationlemCloseTosA} where $A=\set{a_1,a_2}$. Let $T_s=\set{v_1,\dots,v_\ell}$. By the assumption of the lemma, $T_s\subseteq N_G(a_1)$ and $T_s\subseteq N_G(a_2)$. 
Lemma~\ref{lem:technicalLemmaCloseToA} establishes that there are at most $|T_s|$ minimal $s,t$-separators close to $sA$, which can be found in time $O(|T_s|\cdot m)$.
The proof of Lemma~\ref{lem:technicalLemmaCloseToA} establishes that $S\in \minlsepst{G}$ is close to $sA$ if and only if $S$ is close to $sv_iA$ for some $v_i\in T_s$. See complete proof in Section~\ref{sec:AppendixProofsEssentialMiinlstseps} of the Appendix.
\def\belongtoCsLem{
		Let $S\in \minlsepst{G}$ such that $S\subseteq N_G(t)$, and let $u\in \nodes(G){\setminus}\set{s,t}$. If $u\notin C_s(G\sminus S)$ then for every $T\in \minlsepst{G}$, it holds that $u\notin C_s(G\sminus T)$.
}
\begin{lemma}
	\label{lem:belongtoCs}
	\belongtoCsLem
\end{lemma}
\eat{
\begin{proof}
	Since $u\notin C_s(G\sminus S)$, then every path from $u$ to $s$ passes through a vertex in $S$. Now, let $T\in \minlsepst{G}{\setminus} \set{S}$. Since $S\subseteq N_G(t)$, then $S\subseteq T\cup C_t(G\sminus T)$. Therefore, every path from a vertex in $S$ to $s$ passes through a vertex in $T$. Consequently, every path from $u$ to $s$, which passes through a vertex in $S$, must also pass through a vertex in $T$. Therefore, $u\notin C_s(G\sminus T)$.
\end{proof}
}
\eat{
\begin{lemma}
	\label{lem:excludedFromCs}
	Let $S\in \minlsepst{G}$ such that $S\subseteq N_G(s)$, and let $u\in \nodes(G){\setminus}\set{s,t}$. If $u\in C_s(G\sminus S)$ then for every $T\in \minlsepst{G}$, it holds that $u\in C_s(G\sminus T)$.
\end{lemma}
\begin{proof}
	Let $T\in \minlsepst{G}{\setminus} \set{S}$. Since $S\subseteq N_G(s)$, then $S\subseteq T\cup C_s(G\sminus T)$. By Theorem~\ref{thm:parallel} (see~\eqref{eq:thmparallel2}), we have that $C_s(G\sminus S)\subseteq C_s(G\sminus T)$. Therefore, $u\in C_s(G\sminus T)$.
\end{proof}
}

By Lemma~\ref{lem:uniqueCloseVertexSet}, a minimal $s,t$-separator $S\in \minlsepst{G}$, such that $S\subseteq N_G(t)$, is unique, and can be found in polynomial time.  An immediate consequence of Lemma~\ref{lem:belongtoCs} is that we can, in polynomial time, test whether there exists a minimal $s,t$-separator $S\in \minlsepst{G}$ such that $u\in C_s(G\sminus S)$, for a distinguished vertex $u\in \nodes(G)$. To do so, we find the unique $S\in \minlsepst{G}$ such that $S\subseteq N_G(t)$. If $u\in C_s(G\sminus S)$ then the answer is clearly yes. Otherwise, by Lemma~\ref{lem:belongtoCs}, it holds that $u\notin C_s(G\sminus T)$ for any $T\in\minlsepst{G}$.

\def\lemNotCsOrCt{
	Let $T_s\in \minlsepst{G}$ where $T_s\subseteq N_G(s)$. Let $D\in \cc(G\sminus T_s)$ where $s\notin D$ and $t\notin D$. Define $T_D\eqdef T_s\cap N_G(D)$. For every $A\subseteq D$ it holds that:
	\[
		\set{\minlsepst{G}: A\subseteq C_s(G\sminus S)}=\mathop{\textstyle\bigcup}_{v\in T_D}\set{S\in \minlsepst{G}: v\in C_s(G\sminus S)}
	\]
}
\begin{lemma}
	\label{lem:lemNotCsOrCt}
	\lemNotCsOrCt
\end{lemma}
To illustrate Lemma~\ref{lem:lemNotCsOrCt}, consider Figure~\ref{fig:illustrationlemNotCsOrCt}, which shows $T_s\in \minlsepst{G}$ where $T_s\subseteq N_G(s)$, $D\in \cc(G\sminus T_s)$, vertex $a\in A\subseteq D$, and $T_D\eqdef N_G(D)$. Observe that $T_D$ is, by definition, an $s,a$-separator of $G$ where $D=C_a(G\sminus T_D)$ (see Figure~\ref{fig:illustrationlemNotCsOrCt}). Let $S\in \minlsepst{G}$. Lemma~\ref{lem:lemNotCsOrCt} establishes that $A\subseteq C_s(G\sminus S)$ if and only if $C_s(G\sminus S)\cap T_D\neq \emptyset$. The complete proof is in Section~\ref{sec:AppendixProofsEssentialMiinlstseps} of the Appendix.
\begin{figure}[H]
	\centering{
	\begin{minipage}{0.4\textwidth}
		\centering
		\resizebox{0.92\textwidth}{!}{
		\begin{tikzpicture}
			\draw[thick] (0,-0.22) ellipse (0.3cm and 1.6cm);
			
			\foreach \y in {-1.5,-1,-0.5,0,0.5,1} {
				\fill[black] (0, \y) circle (2pt);
			}
			
			\node[fill=black, circle, inner sep=2pt, label=left:$s$] (s) at (-1,0) {};
			
			\foreach \y in {-1.5,-1,-0.5,0,0.5,1} {
				\draw (s) -- (0, \y);
			}
			
			\node at (0.05,-2.1) {$T_s$};
			
			\draw[thick] (3, -0.3) ellipse (1.5cm and 1cm);
			
			\node at (3, -1.8) {$C_t(G \sminus T_s)$};
			
		
			\coordinate (p1) at (2.25, 0.57);
			\coordinate (p2) at (1.79, 0.29);
			\coordinate (p3) at (1.53, -0.09);
			\coordinate (p4) at (1.53, -0.51);
			\coordinate (p5) at (1.79, -0.89);
			\coordinate (p6) at (2.25, -1.17);
	
			\eat{
			\draw (0, 1) -- (p1);
			\draw (0, 0.5) -- (p2);
			\draw (0, 0) -- (p3);
			\draw (0, -0.5) -- (p4);
			\draw (0, -1) -- (p5);
			\draw (0, -1.5) -- (p6);
		}
			
			\node[fill=black, circle, inner sep=2pt, label=right:{\footnotesize $a_1$}] (a1) at (1.4,1.8) {};
			
			\coordinate (q1) at (0.91, 1.70);
			\coordinate (q2) at (0.97, 1.56);
			\coordinate (q3) at (1.08, 1.44);
			
	
			\draw (a1) -- (0,1);
			\draw (a1) -- (0,0.5);
			\draw (a1) -- (0,0);
			\draw (a1) -- (0,-0.5);
			\draw (a1) -- (0,-1);
			\draw (a1) -- (0,-1.5);
			
			
			
			\node[fill=black, circle, inner sep=2pt, label=right:$t$] (s) at (4.5, -0.3) {};
			
				\node[fill=black, circle, inner sep=2pt, label=right:$a_2$] (a2) at (1.8, -0.3) {};
				
				\draw (0, 1) -- (a2);
			\draw (0, 0.5) -- (a2);
			\draw (0, 0) -- (a2);
			\draw (0, -0.5) -- (a2);
			\draw (0, -1) -- (a2);
			\draw (0, -1.5) -- (a2);
			
		\end{tikzpicture}
	}
	\caption{Illustration--Lemma~\ref{lem:technicalLemmaCloseToA}.\label{fig:illustrationlemCloseTosA}}
\end{minipage}
\hfill
	\begin{minipage}{0.4\textwidth}
	\centering
		\resizebox{0.85\textwidth}{!}{
	\begin{tikzpicture}
		\draw[thick] (0,-0.22) ellipse (0.3cm and 1.6cm);
		
		\foreach \y in {-1.5,-1,-0.5,0,0.5,1} {
			\fill[black] (0, \y) circle (2pt);
		}
		
		\node[fill=black, circle, inner sep=2pt, label=left:$s$] (s) at (-1,0) {};
		
		\foreach \y in {-1.5,-1,-0.5,0,0.5,1} {
			\draw (s) -- (0, \y);
		}
		
		\node at (0.05,-2.1) {$T_s$};
		
		\draw[thick] (3, -0.3) ellipse (1.5cm and 1cm);
		
	\node at (3, -1.8) {$C_t(G \sminus T_s)$};
		
		\coordinate (p1) at (2.25, 0.57);
	\coordinate (p2) at (1.79, 0.29);
	\coordinate (p3) at (1.53, -0.09);
	\coordinate (p4) at (1.53, -0.51);
	\coordinate (p5) at (1.79, -0.89);
	\coordinate (p6) at (2.25, -1.17);
		
		\draw (0, 1) -- (p1);
		\draw (0, 0.5) -- (p2);
		\draw (0, 0) -- (p3);
		\draw (0, -0.5) -- (p4);
		\draw (0, -1) -- (p5);
		\draw (0, -1.5) -- (p6);
		
		\draw[thick] (1.5, 1.8) ellipse (0.6cm and 0.5cm);
		
		\node[fill=none,  inner sep=0pt, label=right:{\footnotesize $D=C_a(G\sminus T_D)$}] at (2.1, 1.8) {};
		\node[fill=black, circle, inner sep=1pt, label=right:{\footnotesize $a$}] (a) at (1.4,1.8) {};
		
		\coordinate (q1) at (0.91, 1.70);
		\coordinate (q2) at (0.97, 1.56);
		\coordinate (q3) at (1.08, 1.44);
		
		\draw (0, 1) -- (q1);
		\draw (0, 0.5) -- (q2);
		\draw (0, 0) -- (q3);
		
		\draw[thick,blue] (0, 0.5) ellipse (0.2cm and 0.8cm);
		
		\node[blue] at (-0.4, 1.3) {$T_D$};
		
		\node[fill=black, circle, inner sep=2pt, label=right:$t$] (s) at (4.5, -0.3) {};
	\end{tikzpicture}
		}
\caption{Illustration--Lemma~\ref{lem:lemNotCsOrCt}.\label{fig:illustrationlemNotCsOrCt}}
\end{minipage}
	}
\end{figure}
\subsection{Minimal $s,t$-separators in AT-Free graphs}
\label{sec:MinlSepsATFree}
In any graph $G$, it holds that $N_G(s)\cap N_G(t) \subseteq S$ for every $S\in \minlsepst{G}$. Therefore, finding a minimum $s,t$-separator in $G$ is equivalent to finding a minimum $s,t$-separator in $G\sminus (N_G(s)\cap N_G(t))$. If $G$ is AT-free, then every induced subgraph of $G$ is AT-free, and hence $G\sminus (N_G(s)\cap N_G(t))$ is AT-free is well. Consequently, we make the assumption that $N_G(s)\cap N_G(t)=\emptyset$. In this Section, we prove some useful properties of minimal separators in AT-free graphs.\eat{ Due to space restrictions, the proofs are deferred to Section~\ref{sec:AppendixProofsEssentialMiinlstseps} of the Appendix.}
\def\lemheirarchical{
	Let $G$ be AT-free, $T_s\in \minlsepst{G}$ where $T_s \subseteq N_G(s){\setminus}N_G[t]$, and $C_1,C_2 \in \cc(G\sminus T_s){\setminus} \set{C_s(G\sminus T_s)}$. Then $N_G(C_1)\subseteq N_G(C_2)$ (or $N_G(C_2)\subseteq N_G(C_1)$).
}
\begin{lemma}
	\label{lem:heirarchical}
	\lemheirarchical
\end{lemma}

\def\corrheirarchical{
	Let $G$ be an AT-free graph, $A\subseteq \nodes(G)$, and $T_s\in \minlsepst{G}$ where $T_s \subseteq N_G(s){\setminus}N_G[t]$, we let:
	\[
	\set{C_1,\dots,C_\ell} \eqdef \set{C\in \cc(G\sminus T_s){\setminus}\set{C_s(G\sminus T_s)}: C\cap A\neq \emptyset}
	\]
	Let $S^*\eqdef \bigcap_{i=1}^\ell N_G(C_i)$. Then:
	\[
	\set{S\in \minlsepst{G}: A\subseteq C_s(G\sminus S)}=\bigcup_{v\in S^*}\set{S\in \minlsepst{G}: Av\subseteq C_s(G\sminus S)}
	\]
}
\eat{
\begin{corollary}
	\label{corr:heirarchical}
	\corrheirarchical
\end{corollary}
\begin{proof}
	If $S\in \minlsepst{G}$ where $Av\subseteq C_s(G\sminus S)$ for some $v\in \nodes(G)$, then clearly $S\in \minlsepst{G}$ where $A\subseteq C_s(G\sminus S)$; this proves one direction of the containment.
	
	Assume wlog that $|N_G(C_1)|\leq |N_G(C_2)|\leq \cdots \leq |N_G(C_\ell)|$. By Lemma~\ref{lem:heirarchical}, we have that $N_G(C_1)\subseteq N_G(C_2)\subseteq \cdots \subseteq N_G(C_\ell)$. Therefore, $S^*=N_G(C_1)$. Observe that if $\ell=1$, then $C_1=C_t(G\sminus T_s)$, in which case $S^*=T_s$.
	
	Suppose, by way of contradiction, that there exists a $S\in \minlsepst{G}$ where $A\subseteq C_s(G\sminus A)$ and $C_s(G\sminus S)\cap S^*=C_s(G\sminus S)\cap N_G(C_1)=\emptyset$. Since $S^* \subseteq T_s \subseteq N_G(s)$, it means that $S^* \subseteq S$. Since $S^*$ is, by definition, an $s,C_1$-separator and $S^*\subseteq S$, then $S$ is an $s,C_1$-separator. Therefore, $C_1\cap C_s(G\sminus S)=\emptyset$. In particular, $(C_1\cap A)\cap C_s(G\sminus S)=\emptyset$. Since, by our assumption, $C_1 \cap A\neq \emptyset$, then $A\not\subseteq C_s(G\sminus A)$; a contradiction.
\end{proof}
}

\def\corrheirarchical2{
	Let $G$ be AT-free, $T_s\in \minlsepst{G}$ where $T_s \subseteq N_G(s){\setminus}N_G[t]$, and $\emptyset \subset A\subseteq \nodes(G)$, such that $A\cap (C_s(G\sminus T_s)\cup T_s\cup C_t(G\sminus T_s))=\emptyset$. Define $\set{C_1,\dots,C_\ell} \eqdef \set{C\in \cc(G\sminus T_s): C\cap A\neq \emptyset}$, and $S^*\eqdef \bigcap_{i=1}^\ell N_G(C_i)$. Then:
	\begin{align*}
	\set{S\in \minlsepst{G}: A\subseteq C_s(G\sminus S)}=\mathop{\textstyle\bigcup}_{v\in S^*}\set{S\in \minlsepst{G}: v\in C_s(G\sminus S)}
	\end{align*}
}
\begin{corollary}
	\label{corr:heirarchical2}
	\corrheirarchical2
\end{corollary}
\eat{
\begin{proof}
	Since $\emptyset \subset A \subseteq \nodes(G){\setminus}(C_s(G\sminus T_s)\cup T_s\cup C_t(G\sminus T_s))$, then $\ell \geq 1$.
	Assume wlog that $|N_G(C_1)|\leq |N_G(C_2)|\leq \cdots \leq |N_G(C_\ell)|$. Since $G$ is AT-free, and $T_s\cap N_G[t]=\emptyset$, then by Lemma~\ref{lem:heirarchical}, we have that $N_G(C_1)\subseteq N_G(C_2)\subseteq \cdots \subseteq N_G(C_\ell)$. Therefore, $S^*=N_G(C_1)$. 
	
	Let $S\in \minlsepst{G}$ where $v\in C_s(G\sminus S)$ for some $v\in S^*$. Since $v\in S^*\eqdef \bigcap_{i=1}^\ell N_G(C_i)$, then by Lemma~\ref{lem:lemNotCsOrCt}, we have that $S\in \minlsepst{G}$ where $C_i \subseteq C_s(G\sminus S)$ for every $i\in \set{1,2,\dots,\ell}$. Since $A\subseteq \bigcup_{i=1}^\ell C_i$, then $A\subseteq C_s(G\sminus S)$.

	Suppose, by way of contradiction, that there exists an $S\in \minlsepst{G}$ where $A\subseteq C_s(G\sminus S)$ and $C_s(G\sminus S)\cap S^*=C_s(G\sminus S)\cap N_G(C_1)=\emptyset$. Since $S^* \subseteq T_s \subseteq N_G(s)$, it means that $S^* \subseteq S$. Since $S^*$ is, by definition, an $s,C_1$-separator and $S^*\subseteq S$, then $S$ is an $s,C_1$-separator. Therefore, $C_1\cap C_s(G\sminus S)=\emptyset$. In particular, $(C_1\cap A)\cap C_s(G\sminus S)=\emptyset$. Since $C_1 \cap A\neq \emptyset$, then $A\not\subseteq C_s(G\sminus S)$; a contradiction.
\end{proof}
}

\section{Finding All Close Minimal $s,t$-Separators in AT-Free Graphs}
\label{sec:MainThm}
In this section, we prove Theorem~\ref{thm:mainThmATFree} that forms the basis of our algorithm. The proof relies on the results established in Section~\ref{sec:charMinlSeps}. 
\begin{reptheorem}{\ref{thm:mainThmATFree}}
	\mainThmATFree
\end{reptheorem}
The procedure $\algname{CloseTo}$ that returns the minimal $s,t$-separators of $G$ that are close to $sA$ in time $O(n^2m)$ (or in time $O(nm)$), and its detailed runtime analysis, is deferred to Section~\ref{sec:CloseToPseudocode} of the Appendix. We now prove Theorem~\ref{thm:mainThmATFree}.
\begin{proof}
	Let $T_t, T_s\in \minlsepst{G}$, where $T_t\subseteq N_G(t)$ and $T_s\subseteq N_G(s)$. By Lemma~\ref{lem:uniqueCloseVertexSet}, $T_s$ and $T_t$ are the unique minimal $s,t$-separators that are close $s$ and $t$, respectively, and they can be found in time $O(m)$. If $A\not\subseteq C_s(G\sminus T_t)$, then by Lemma~\ref{lem:belongtoCs}, it holds that $A\not\subseteq C_s(G\sminus S)$ for every $S\in \minlsepst{G}$. Hence, if $A\not\subseteq C_s(G\sminus T_t)$, then there are no minimal $s,t$-separators close to $sA$ (i.e., $\F_{sA}(G)=\emptyset$). So, we assume that $A\subseteq C_s(G\sminus T_t)$.
	For every $S\in \minlsepst{G}$ where $A\subseteq C_s(G\sminus S)$ it must hold that $N_G(sA)\cap N_G(t) \subseteq S$. Since $G\sminus (N_G(sA)\cap N_G(t))$ is an induced subgraph of an AT-free graph $G$, then $G\sminus (N_G(sA)\cap N_G(t))$ is AT-free as well. Therefore, we assume that $G$ is AT-free, and that $N_G(sA)\cap N_G(t)=\emptyset$. 
	
	Let $A_1\eqdef A\cap (C_s(G\sminus T_s)\cup T_s\cup C_t(G\sminus T_s))$, and $A_2\eqdef A{\setminus}A_1$. Let $\set{C_1,\dots,C_\ell}\eqdef\set{C\in \cc(G\sminus T_s): C\cap A_2 \neq \emptyset}$, and $S^*\eqdef\bigcap_{i=1}^\ell N_G(C_i)$.\eat{If $A_2\neq \emptyset$ then} By Corollary~\ref{corr:heirarchical2}:
	\begin{equation}
		\set{S\in \minlsepst{G}: A_2\subseteq C_s(G\sminus S)}=\mathop{\textstyle\bigcup}_{v\in S^*}\set{S\in \minlsepst{G}: v\in C_s(G\sminus S)} \label{eq:mainThm11}
	\end{equation}
	Therefore, we have that:
	\begin{align}
		&\set{S\in \minlsepst{G}: A\subseteq C_s(G\sminus S)}\\
		&=\set{S\in \minlsepst{G}: A_1\subseteq C_s(G\sminus S)}\cap \set{S\in \minlsepst{G}: A_2\subseteq C_s(G\sminus S)} \nonumber \\
		&\underbrace{=}_{\eqref{eq:mainThm11}}\set{S\in \minlsepst{G}: A_1\subseteq C_s(G\sminus S)}\cap\mathopen{}\big(\mathop{\textstyle\bigcup}_{v\in S^*}\set{S\in \minlsepst{G}: v\in C_s(G\sminus S)}\big) \mathclose{} \nonumber \\
		&=\mathop{\textstyle\bigcup}_{v\in S^*}\set{S\in \minlsepst{G}: A_1v\in C_s(G\sminus S)} \label{eq:mainThm22}
	\end{align}
	Since $S^* \subseteq T_s$, then $A_1v\subseteq C_s(G\sminus T_s)\cup T_s\cup C_t(G\sminus T_s)$. 
	Therefore, if we show that for every $v\in S^*$, there are at most $n$ minimal $s,t$-separators that are close to $svA_1$ (i.e., $|\F_{svA_1}|\leq n$), which can be found in time $O(nm)$, then we get that there are at most $|S^*|\cdot n \leq n^2$ minimal $s,t$-separators that are close to $sA$, which can be found in time $O(n^2m)$.
	Overall, to prove the claim we need to show that if $A\subseteq C_s(G\sminus S)\cup T_s \cup C_t(G\sminus T_s)$, then there are at most $n$ minimal $s,t$-separators that are close to $sA$, which can be found in time $O(nm)$. The rest of the proof is devoted to this setting.
	\newline

	\noindent {\bf Claim 1}: $\set{S\in \minlsepst{G}: A\subseteq C_s(G\sminus S)}\subseteq \minlsep{sA,t}{G}$. \newline
	\noindent {\bf Proof.} Let $S\in \minlsepst{G}$ where $A\subseteq C_s(G\sminus S)$. Then $S$ is an $sA,t$-separator. By Lemma~\ref{lem:fullComponents}, $S=N_G(C_t(G\sminus S))=N_G(C_s(G\sminus S))$. By Lemma~\ref{lem:simpAB}, $S\in \minlsep{sA,t}{G}$.
	\qed
	
	By Lemma~\ref{lem:MinlsASep}, we have that $\minlsep{sA,t}{G}=\minlsepst{H}$ where $\nodes(H)=\nodes(G)$ and $\edges(H)=\edges(G)\cup \set{(s,a): a\in N_G[A]}$. 
	Let $S_1\in \minlsepst{H}$ where $S_1 \subseteq N_H(s)$. By Lemma~\ref{lem:uniqueCloseVertexSet}, $S_1$ is unique and can be found in time $O(m)$. In addition, $S_1\in \minlsep{sA,t}{G}$. Let $C_s,C_t \in \cc(G\sminus S_1)$ be the connected components of $G\sminus S_1$ that contain vertices $s$ and $t$, respectively.
	\newline
	
	\noindent {\bf Claim 2}: For every $S\in \minlsepst{G}$: if $A\subseteq C_s(G\sminus S)$, then $C_s(G\sminus S_1)\subseteq C_s(G\sminus S)$. \newline
	\noindent {\bf Proof.}  Let $S\in \minlsepst{G}$ where $A\subseteq C_s(G\sminus S)$. By Claim 1, it holds that $S\in \minlsep{sA,t}{G}=\minlsepst{H}$. Since $S_1\in \minlsepst{H}$ where $S_1\subseteq N_H(s)$, then $S_1\subseteq S\cup C_s(H\sminus S)$, and by Lemma~\ref{lem:inclusionCsCt}, that $C_s(H\sminus S_1)\subseteq C_s(H\sminus S)$. Since $S\in \minlsepst{G}$ where $A\subseteq C_s(G\sminus S)$, then $C_s(H\sminus S)=C_s(G\sminus S)$. Since $\edges(G)\subseteq \edges(H)$, then:
	\[
		C_s(G\sminus S_1)\subseteq C_s(H\sminus S_1)\underbrace{\subseteq}_{\text{Lemma}~\ref{lem:inclusionCsCt}} C_s(H\sminus S)\underbrace{=}_{A\subseteq C_s(G\sminus S)}C_s(G\sminus S).
	\]

	\noindent {\bf Claim 3}: $A\cap(T_s\cup C_s(G\sminus T_s))\subseteq C_s(G\sminus S_1)$. \newline
	\noindent {\bf Proof.} Since $T_s\subseteq N_G(s)$, then $T_s\subseteq S_1\cup C_s(G\sminus S_1)$. By Lemma~\ref{lem:inclusionCsCt}, $C_s(G\sminus T_s)\subseteq C_s(G\sminus S_1)$. 
	Since $S_1\in \minlsep{sA,t}{G}$, then $A\cap S_1=\emptyset$. Since $A\cap T_s\subseteq N_G(s)$, then $A\cap T_s \subseteq C_s(G\sminus S_1)$.
	\qed 

	Consider the graph $G\sminus S_1$. There are two cases: $A\subseteq C_s(G\sminus S_1)$ and $A\not\subseteq C_s(G\sminus S_1)$. Recall that $C_s\eqdef C_s(G\sminus S_1)$, where $S_1\in \minlsepst{H}$ and $S_1\subseteq N_H(s)$. \newline
	
	\noindent {\bf Case 1: $A\subseteq C_s$.} Since $S_1\in \minlsepst{H}=\minlsep{sA,t}{G}$ and $A\subseteq C_s$, then by Lemma~\ref{lem:simpAB} it holds that $S_1=N_G(C_s)\cap N_G(C_t)$. By Lemma~\ref{lem:fullComponents}, we have that $S_1\in \minlsepst{G}$. We claim that $S_1$ is the unique minimal $s,t$-separator that is close to $sA$. Let $S\in \minlsepst{G}$ such that $A\subseteq C_s(G\sminus S)$. By Claim 2, $C_s(G\sminus S_1)\subseteq C_s(G\sminus S)$. Hence, for this case the Theorem is proved.

	\noindent {\bf Case 2: $A\not\subseteq C_s$.} Let $A'\eqdef A{\setminus}C_s$. By Claim 3, we have that $A'\subseteq \nodes(G){\setminus}(C_s(G\sminus T_s)\cup T_s)$. Since $A\subseteq C_s(G\sminus T_s)\cup T_s \cup C_t(G\sminus T_s)$, we have that $A'\subseteq C_t(G\sminus T_s)$. 
	Therefore, for every $a\in A'$ there is an $a,t$-path in $G$ that resides entirely in $C_t(G\sminus T_s)$, and hence avoids $N_G[s]$. By our assumption that $A' \subseteq A\subseteq C_s(G\sminus T_t)$, there is an $s,a$-path in $G$ that resides entirely in $C_s(G\sminus T_t)$, and hence avoids $N_G[t]$. 
	
	Define $Q_s \eqdef N_G(C_s)$. By definition, $Q_s$ is an $s,t$-separator, and $Q_s \subseteq S_1$, and hence $Q_s=N_G(C_s)\cap S_1$.
	Since $S_1\in \minlsep{sA,t}{G}$, then by Lemma~\ref{lem:simpAB}, it holds that $S_1\subseteq N_G(C_t)$. Therefore, $Q_s$ is an $s,t$-separator where $Q_s\subseteq N_G(C_s)\cap N_G(C_t)$. By Lemma~\ref{lem:fullComponents}, $Q_s\in \minlsepst{G}$; see Figure~\ref{fig:mainThmCase2} for illustration. \newline

	\noindent {\bf Claim 4}: For every $a\in A'$, it holds that $Q_s \subseteq N_G(a)$. \newline
	\noindent {\bf Proof.} Suppose, by way of contradiction, that $Q_s\not\subseteq N_G(a)$ for some $a\in A'$, and let $v\in Q_s{\setminus}N_G(a)$. By Definition, $v\in Q_s\subseteq N_G(C_s)\cap N_G(C_t)$. Therefore, there is an $s,t$-path in $G$ that passes through $v$, denoted $P^v_{s,t}$, that resides entirely in $C_s \cup \set{v}\cup C_t$ (see Fig.~\ref{fig:mainThmCase2}). Since $C_s\eqdef C_s(G\sminus S_1)$, where  $a\notin S_1 \cup C_s$, and $C_t\eqdef C_t(G\sminus S_1)$ where $a\notin C_t\cup S_1$, then $N_G[a]\cap (C_s\cup C_t)=\emptyset$. Combined with the assumption that $v\notin N_G[a]$, we get that $N_G[a]\cap(C_s\cup \set{v}\cup C_t)=\emptyset$, and hence $N_G[a]\cap \nodes(P^v_{s,t})=\emptyset$. Therefore, there is an $s,t$-path in $G$ (via $v$) that avoids $N_G[a]$. Since  there exists an $s,a$-path in $G$ that avoids $N_G[t]$ (i.e., $P_{a,s}$) and an $a,t$-path in $G$ that avoids $N_G[s]$ (i.e., $P_{a,t}$), we get that $s,a,t$ form an asteroidal triple in $G$ (see Fig.~\ref{fig:mainThmCase2}). But this is a contradiction. Therefore, $Q_s\subseteq N_G(a)$.
	\qed 
	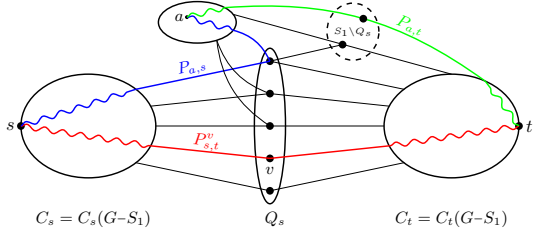
\begin{figure}[h]
		\centering
		 \resizebox{0.6\textwidth}{!}{
		\begin{tikzpicture}
			
			\draw[thick] (0, 0) ellipse [x radius=0.3cm, y radius=1.5cm];
			
			\foreach \i in {0,1,2,3,4} {
				\node[circle, fill=black, inner sep=1.5pt] at (0, 1.25-0.625*\i) {};  
			}
			
			\node at (0.1, -1.8) {\footnotesize{$Q_s$}};
			
			\draw[thick] (-3.5, 0) ellipse [x radius=1.3cm, y radius=1cm];
			\node[circle, fill=black, inner sep=1.5pt] at (-4.8,0) {};  
			\node at (-5, 0) {$s$};
			
			\node at (-3.4, -1.8) {\footnotesize{$C_s=C_s(G\sminus S_1)$}};
			
			\draw[thick] (3.5, 0) ellipse [x radius=1.3cm, y radius=1cm];
			
			\node at (3.5, -1.8) {\footnotesize{$C_t=C_t(G\sminus S_1)$}};
				\node[circle, fill=black, inner sep=1.5pt] at (4.8,0) {};  
			\node at (5, 0) {$t$};
		
			\draw[thick] (-1.4, 2.0) circle [x radius=0.75cm, y radius=0.4cm];
			\coordinate (a) at (-1.6,2.1);
			\node[circle, fill=black, inner sep=0.75pt] at (a) {}; 
			\node at  (-1.75,2.1) {\footnotesize{$a$}};
			
			\node at  (0,-0.85) {\footnotesize{$v$}};

			\coordinate (p1) at (-2.590,0.707);
			\coordinate (p2) at  (-2.297,0.383);
			\coordinate (p3) at (-2.2,0);
			\coordinate (p4) at  (-2.297,-0.383);
			\coordinate (p5) at (-2.590,-0.707);

			\coordinate (q1) at (2.590,0.707);
			\coordinate (q2) at  (2.297,0.383);
			\coordinate (q3) at (2.2, 0);
			\coordinate (q4) at  (2.297,-0.383);
			\coordinate (q5) at (2.590,-0.707);

			\draw[thick, blue] (p1) -- (0,1.25);
			\draw[] (p2) -- (0,0.625);
			\draw (p3) -- (0,0);
			\draw[thick, red] (p4) -- (0,-0.625);
			\draw[] (p5) -- (0,-1.25);
			
			\draw (q1) -- (0,1.25);
			\draw[] (q2) -- (0,0.625);
			\draw (q3) -- (0,0);
			\draw[thick, red] (q4) -- (0,-0.625);
			\draw[] (q5) -- (0,-1.25);
			
			\coordinate (a1) at  (-0.86975,2.2828);
			\coordinate (a2) at  (-0.707,2.1532);
			
			\coordinate (a3) at  (-0.7071,1.8469);
			\coordinate (a4) at  (-1.025,1.6536);
			
			\coordinate (q6) at (4.15,0.866);
		
			  \draw (a2) -- (3.5,1.0);
			  \draw[bend left=20, thick, green] (a1) to (q6); 
			
			  \coordinate (x1) at (1.8,2.07);
			  \coordinate (x2) at (1.4,1.57);
			  \node[circle, fill=black, inner sep=1.5pt] at (x1) {};
			   \node[green, rotate=-23] at  (2.7,1.93) {\footnotesize{$P_{a,t}$}};
			  
			  \node[blue, rotate=12] at  (-1.5,1.13) {\footnotesize{$P_{a,s}$}};
			   \node[red, rotate=-5] at  (-1.22,-0.31) {\footnotesize{$P_{s,t}^v$}};
			  	
			  \node[circle, fill=black, inner sep=1.5pt] at (x2) {};
			  	
			  \draw (x2) -- (0,1.25);
			  
			  \draw[thick, blue, bend left=25] (a3) to (0,1.25);
			  \draw[bend right=25] (a4) to (0,0.625);
			   \draw[bend right=25] (a4) to (0,0);

			  \draw[thick, red, decorate, decoration={snake, amplitude=0.5mm}] (-4.8,0) -- (p4);
			  \draw[thick, red, decorate, decoration={snake, amplitude=0.5mm}] (4.8,0) -- (q4);

			   \draw[thick, blue, decorate, decoration={snake, amplitude=0.5mm}] (a) -- (a3);
			   \draw[thick, blue, decorate, decoration={snake, amplitude=0.5mm}] (-4.8,0) -- (p1);
			    
			   \draw[thick, green, decorate, decoration={snake, amplitude=0.5mm}] (a) -- (a1);
			   \draw[thick, green, decorate, decoration={snake, amplitude=0.5mm}] (4.8,0) -- (q6);
			   
			   \draw[thick,dashed] (1.6, 1.82) ellipse [x radius=0.5cm, y radius=0.55cm];
			   \node at (1.6,1.82) {\tiny{$S_1{\setminus}Q_s$}};
		\end{tikzpicture}
	}
	\caption{Illustration for the proof of Theorem~\ref{thm:mainThmATFree} (Case 2). \label{fig:mainThmCase2}}
	\end{figure}

	Let $S\in \minlsepst{G}$ where $A\subseteq C_s(G\sminus S)$. By Claim 4, $Q_s \subseteq N_G(a)$ for every $a\in A'\subseteq A$. Since $A'\neq \emptyset$, then $Q_s \subseteq S\cup C_s(G\sminus S)$. Since $S,Q_s\in \minlsepst{G}$, then by Lemma~\ref{lem:inclusionCsCt}, we have that $C_s(G\sminus Q_s)\subseteq C_s(G\sminus S)$. Therefore, we get that:
	\begin{align}
		&\set{S\in \minlsepst{G}: A \subseteq C_s(G\sminus S)}\subseteq \set{S\in \minlsepst{G}: C_s(G\sminus Q_s)\subseteq C_s(G\sminus S)} 	\label{eq:MainThm1}
	\end{align}
Since $C_s(G\sminus Q_s)$ is a connected component, then by Lemma~\ref{lem:contract}, we have that:
\begin{align}
	 \set{S\in \minlsepst{G}: C_s(G\sminus Q_s)\subseteq C_s(G\sminus S)}  =\minlsepst{M} \label{eq:MainThm2}
		\end{align}
\eat{
	\begin{align}
	&\set{S\in \minlsepst{G}: A \subseteq C_s(G\sminus S)}\subseteq \set{S\in \minlsepst{G}: S\subseteq Q_s \cup C_t(G\sminus Q_s)} 	\label{eq:MainThm1}\\
	&\text{By Lemma~\ref{lem:addEdgesFromSTot}, it holds that: }
	\set{S\in \minlsepst{G}: S\subseteq Q_s \cup C_t(G\sminus Q_s)} =\minlsepst{M} 	\label{eq:MainThm2}
	\end{align}
}
	where $M$ is the graph that results from $G$ by contracting $C_s(G\sminus Q_s)$ to vertex $s$. 
	Also, by Lemma~\ref{lem:contract}, we have that $C_s(G\sminus S)=C_s(M\sminus S)\cup C_s(G\sminus Q_s)$, and $C_t(G\sminus S)=C_t(M\sminus S)$ for every $S\in \minlsepst{M}$. 
	Let $D\eqdef A{\setminus}C_s(G\sminus Q_s)$.\newline
	
	\noindent {\bf Claim 5}: $\F_{sD}(M)=\F_{sA}(G)$.\newline
	The technical proof of this claim is deferred to Section~\ref{sec:CloseToPseudocode} of the Appendix.

	Since $\F_{sD}(M)=\F_{sA}(G)$, we are left to show that $|\F_{sD}(M)|\leq n$, and that $\F_{sD}(M)$ can be computed in time $O(nm)$.
	By definition of contraction, we have that $N_M(s) \supseteq N_G(C_s(G\sminus Q_s))=Q_s$, and that $C_t(M\sminus Q_s)=C_t(G\sminus Q_s)$. Consequently, we have that $Q_s\in \minlsepst{M}$ where $Q_s\subseteq N_M(s)$. By Lemma~\ref{lem:uniqueCloseVertexSet}, $Q_s$ is the unique minimal $s,t$-separator of $M$ that is close to $s$. For every $a\in A$, either $a\in C_s(G\sminus S_1)\subseteq C_s(G\sminus Q_s)$, or by Claim 4, $Q_s \subseteq  N_G(a) \subseteq N_M(a)$. By Lemma~\ref{lem:technicalLemmaCloseToA}, there are at most $|Q_s|\leq n$ minimal $s,t$-separators that are close to $sD$ in $M$, and they can be found in time $O(|Q_s|\cdot m)=O(nm)$.
\end{proof}

\section{Conclusion}
\label{sec:conclusion}
In this paper, we presented the first polynomial-time algorithm to find a connectivity-preserving, minimum-weight  A,B -separator in AT-free graphs, a general class encompassing interval, cocomparability, cobipartite, and trapezoid graphs. Our algorithm leverages key properties of minimal separators in AT-free graphs for an efficient solution. To our knowledge, this is also the first polynomial-time algorithm to find a connectivity-preserving  A,B -separator when  A  and  B  are unbounded in any non-trivial, infinite graph class. Additionally, our results on minimal separators in AT-free graphs may be of independent interest, offering insights applicable to other problems.
\eat{
In this paper, we presented the first polynomial-time algorithm that finds a connectivity-preserving, minimum-weight $A,B$-separator in AT-free graphs. AT-free graphs are a fairly general class, encompassing several important subclasses such as interval graphs, cocomparability graphs, cobipartite graphs, and trapezoid graphs. Our algorithm effectively leverages key properties of minimal separators in AT-free graphs, enabling an efficient solution to this problem. To the best of our knowledge, this is also the first algorithm that, in polynomial time, finds a connectivity-preserving $A,B$-separator when $A$ and $B$ are unbounded in size, in any non-trivial and infinite class of graphs. \eat{Despite the complexity of the problem, the algorithm remains straightforward to implement, as demonstrated by the pseudocode provided in the paper.}
In addition, our results on minimal separators in AT-free graphs may be of independent interest, offering new insights that could be applicable to other problems in this area.
}
\eat{
As part of future work, we plan to extend our investigation to other graph classes, focusing on characterizing the set of close minimal separators in these classes. This would not only broaden the applicability of our algorithm but also potentially uncover new structural properties that could lead to efficient solutions for related connectivity problems in more complex graph families.
}

\newpage 

\bibliography{main}

\newpage 

\appendix

\section{Minimal Separators Between Vertex-Sets - Proofs from Section~\ref{sec:minlsepsVertexSets}}
\label{sec:minsepsvertexsets}
\eat{
Let $A$ and $B$ be two disjoint, non-adjacent subsets of $\nodes(G)$. A vertex-set $X\subseteq \nodes(G)$ is called an \e{$A,B$-separator} if, in the graph $G\sminus X$, there is no path between $A$ and $B$. We say that $X$ is a minimal $A,B$-separator if no proper subset of $X$ has this property. We say that a subset $X\subseteq \nodes(G)\setminus AB$ is a minimum $A,B$-separator if, for every $A,B$-separator $S$, it holds that $|X|\leq |S|$. We denote by $\minsep_{A,B}(G)$ the set of minimum $A,B$-separators, and by $\kappa_{A,B}(G)$ the size of a minimum $A,B$-separator.
We denote by $\minlsep{A,B}{G}$  the set of minimal $A,B$-separators in $G$. 
\eat{
We say that an $A,B$-separator $X$ is 
\e{safe} if there are two distinct, connected components $C_A,C_B\in \cc_G(X)$, where $A\subseteq C_A$ and $B\subseteq C_B$.}
In this short section, we establish that finding a minimal or minimum $A,B$-separator, where $A,B\subseteq \nodes(G)$ are disjoint and non-adjacent, can be reduced to finding a minimal or minimum $a,b$-separator in $G'$, where $a\in A,b\in B$, and $G'$ is the graph that results from merging $A$ to $a$ and $B$ to $b$ (see~\eqref{eq:mergeDef}).
}

\begin{replemma}{\ref{lem:simpAB}}
	\simpABlemma
\end{replemma}
\begin{proof}
	If $S\in \minlsep{A,B}{G}$, then for every $w\in S$ it holds that $S{\setminus} \set{w}$ no longer separates $A$ from $B$. Hence, there is a path from some $a\in A$ to some $b\in B$ in $G\sminus (S{\setminus} \set{w})$. 
	Let $C_a$ and $C_b$ denote the connected components of $\cc(G\sminus S)$ containing $a\in A$ and $b\in B$, respectively. Since $C_a$ and $C_b$ are connected in $G\sminus (S{\setminus} \set{w})$, then $w\in N_G(C_a)\cap N_G(C_b)$.
	
	Suppose that for every $w\in S$, there exist two connected components $C_A,C_B\in \cc_G(S)$ such that $C_A\cap A\neq \emptyset$, $C_B\cap B\neq \emptyset$, and $w\in N_G(C_A)\cap N_G(C_B)$. If $S\notin \minlsep{A,B}{G}$, then $S{\setminus} \set{w}$ separates $A$ from $B$ for some $w\in S$. Since $w$ connects $C_A$ to $C_B$ in $G\sminus (S{\setminus} \set{w})$, no such $w\in S$ exists, and thus $S\in \minlsep{A,B}{G}$.
\end{proof}

Observe that Lemma~\ref{lem:simpAB} implies Lemma~\ref{lem:fullComponents}. By Lemma~\ref{lem:simpAB}, it holds that $S\in \minlsepst{G}$ if and only if $S$ is an $s,t$-separator and $S\subseteq N_G(C_s(G\sminus S))\cap N_G(C_t(G\sminus S))$. By definition, $N_G(C_s(G\sminus S))\subseteq S$ and $N_G(C_t(G\sminus S))\subseteq S$, and hence $S=N_G(C_s(G\sminus S))\cap N_G(C_t(G\sminus S))$, and $S=N_G(C_s(G\sminus S))=N_G(C_t(G\sminus S))$.

\begin{lemma}
	\label{lem:minlsepsupergraph}
	Let $G$ and $H$ be graphs where $\nodes(G)=\nodes(H)$ and $\edges(G)\subseteq \edges(H)$. Let $S\in \minlsep{A,B}{G}$. If $S$ is an $A,B$-separator in $H$, then $S\in \minlsep{A,B}{H}$.
\end{lemma}
\begin{proof}
	Since $S\in \minlsep{A,B}{G}$, then by Lemma~\ref{lem:simpAB}, for every $w\in S$ there exist $C_A^w(G\sminus S)\in \cc(G\sminus S)$ and $C_B^w(G\sminus S) \in \cc(G\sminus S)$ where $A\cap C_A^w(G\sminus S) \neq \emptyset$, $B\cap C_B^w(G\sminus S)\neq \emptyset$, and $w\in N_G(C_A^w(G\sminus S))\cap N_G(C_B^w(G\sminus S))$. Since $\edges(H)\supseteq \edges(G)$, and since $S$ is an $A,B$-separator in $H$, then $C_A^w(H\sminus S)\supseteq C_A^w(G\sminus S)$, and $C_B^w(H\sminus S)\supseteq C_B^w(G\sminus S)$. Therefore, $w\in N_H(C_A^w(H\sminus S))\cap N_H(C_B^w(H\sminus S)$) for every $w\in S$.
By Lemma~\ref{lem:simpAB}, it holds that $S\in \minlsep{A,B}{H}$.
\end{proof}

\begin{replemma}{\ref{lem:MinlsASep}}
	\lemMinlsASep
\end{replemma}
\begin{proof}
	Let $T \in \minlsep{sA,t}{G}$, and let $C_1,\dots,C_k$ denote the connected components of $\cc(G\sminus T)$ containing vertices from $sA$, and let $C_t \in \cc(G\sminus T)$ denote the connected component of $\cc(G\sminus T)$ that contains $t$. Assume wlog that $s\in C_1$. By definition, the edges added to $G$ to form $H$ are between $C_1$ and $C_1\cdots C_k \cup T$. Therefore, $T$ separates $sA$ from $t$ in $H$, and in particular, $T$ separates $s$ from $t$ in $H$. Since $\edges(H) \supseteq \edges(G)$, then by Lemma~\ref{lem:minlsepsupergraph}, if $T \in \minlsep{sA,t}{G}$ and $T$ is an $sA,t$-separator in $H$, then $T\in \minlsep{sA,t}{H}$. 
	Since, by construction, $A\subseteq N_{H}[s]{\setminus} T$ then $H\sminus T$ contains two connected components $C_{sA}$ that contains $sA$, and $C_t$ that contains $t$. By Lemma~\ref{lem:simpAB}, we have that $T = N_{H}(C_{sA})\cap N_{H}(C_t)$. By Lemma~\ref{lem:fullComponents}, we have that $T\in \minlsepst{H}$.
	
	Let $T\in \minlsepst{H}$. We first show that $T$ separates $sA$ from $t$ in $G$; if not, there is a path from $x\in sA$ to $t$ in $G\sminus T$. By definition of $H$, $x\in N_{H}[s]{\setminus}T$. This means that there is a path from $s$ to $t$ (via $x$) in $H \sminus T$, which is a contradiction. If $T \notin \minlsep{sA,t}{G}$, then there is a $T' \in \minlsep{sA,t}{G}$ where $T' \subset T$. By the previous direction,  $T'\in \minlsep{sA,t}{G} \subseteq \minlsepst{H}$, and hence $T'\in \minlsepst{H}$, contradicting the minimality of $T\in \minlsepst{H}$.
\end{proof}

\section{Proofs from Section~\ref{sec:MaintainingMinlseps}} 
\label{sec:proofsFromPrelims}

\begin{replemma}{\ref{lem:contract}}
	\lemContract
\end{replemma} 
\begin{proof}
	We prove for the case where $A=\set{u}$ and $u\in N_G(s)$. The claim then follows by a simple inductive argument on $|A|$, the cardinality of $A$, by noticing that since $G[sA]$ is connected, then $A\cap N_G(s)\neq \emptyset$.
	
	Let $S\in \minlsepst{G}$ where $u\in C_s(G\sminus S)$. This means that $N_G[u] \subseteq C_s(G\sminus S) \cup S$. By definition of contraction, $\edges(H){\setminus}\edges(G)\subseteq \set{(s,v):v\in N_G[u]} \subseteq C_s(G\sminus S)\cup S$. In other words, every edge in $\edges(H){\setminus}\edges(G)$ is between $s$ and a vertex in $S\cup C_s(G\sminus S)$. Therefore, $S$ is an $s,t$-separator in $H$. For the same reason, we have that $C_t(H\sminus S)=C_t(G\sminus S)$, and in particular that $G[C_t(G\sminus S)\cup S]=H[C_t(H\sminus S)\cup S]$. Therefore, $S=N_G(C_t(G\sminus S))=N_H(C_t(H\sminus S))$.
	We claim that $S\subseteq N_H(C_s(H\sminus S))$. Since $S\in \minlsepst{G}$, then by Lemma~\ref{lem:fullComponents}, $S=N_G(C_s(G\sminus S))$. Take any $v\in S$, and let $x\in N_G(v)\cap C_s(G\sminus S)$. If $x\in C_s(H\sminus S)$, then $x\in N_H(v)\cap C_s(H\sminus S)$, and hence $v\in N_H(C_s(H\sminus S))$. Otherwise, $x=u$, and by definition of contraction, we get that $s\in N_H(v)$. Therefore, $S\subseteq N_H(C_s(H\sminus S))$. So, we have that $S$ is an $s,t$-separator of $H$ where $S\subseteq N_H(C_t(H\sminus S))\cap N_H(C_s(H\sminus S))$. By Lemma~\ref{lem:fullComponents}, $S\in \minlsepst{H}$.
	
	Now, let $S\in \minlsepst{H}$. 
	Since $u\notin \nodes(H)$, then $u\notin S$. Let $C_s,C_t \in \cc(H\sminus S)$ be the full connected components associated with $S$ in $H$ that contain $s$ and $t$ respectively. That is, $N_{H}(C_s)=N_{H}(C_t)=S$. 
	We claim that $C_s(G\sminus S)=C_s \cup \set{u}$.
	First, we show that $G[C_s \cup \set{u}]$ is connected. To see why, take any vertex $x\in C_s$. If there is no $s,x$-path in $G[C_s\cup \set{u}]$, it means that every $s,x$-path in $C_s(H\sminus S)$ uses an edge $(s,w)\in \edges(H){\setminus}\edges(G)\subseteq \set{(s,v):v\in N_G(u)}$. But then, $G[C_s\cup \set{u}]$ contains the subpath $s-u-w$, which means that $s$ and $x$ are connected in $G[C_s\cup \set{u}]$. Since all vertices in $C_s\cup \set{u}$ are connected to $s$ in $G$, and since $(C_s\cup \set{u})\cap S=\emptyset$, we get that $C_s \cup \set{u}$ is connected in $G\sminus S$. Therefore, $C_s\cup \set{u} \subseteq C_s(G\sminus S)$. For the other direction, take $x\in C_s(G\sminus S)$, and let $P_{s,x}$ be an $s,x$ path in $C_s(G\sminus S)$. If $u\notin \nodes(P_{s,x})$, then by definition of contraction, $P_{s,x}$ is an $s,x$ path in $H$ that avoids $S$, and hence $x\in C_s(H\sminus S)$. If $u\in \nodes(P_{s,x})$, then let $u'$ be the vertex that immediately follows $u$ on the path $P_{s,x}$. By definition of contraction, $(s,u')\in \edges(H)$. So, we have an $s,x$-path in $H$, via $u'$, that avoids $S\cup \set{u}$, and hence $x\in C_s(H\sminus S)$. Overall, we showed that $C_s(G\sminus S)=C_s(H\sminus S)\cup \set{u}$. 
	
	We claim that $S=N_G(C_s \cup \set{u})$. Since $C_s \cup \set{u}$ is a connected component of $G\sminus S$, then $N_G(C_s \cup \set{u})\subseteq S$. Now, take $v\in S$. Then $v\in N_{H}(x)$ for some vertex $x\in C_s$. If $v\in N_G(x)$, then $v\in N_G(C_s)$, and we are done. Otherwise, $x=s$ because all edges in $\edges(H){\setminus}\edges(G)$ have an endpoint in $s \in C_s$. Since $(s,v) \in \edges(H){\setminus}\edges(G)$, then $v\in N_G(u)$. Therefore, $v\in N_G(C_s\cup \set{u})$. So, we get that $S=N_G(C_s\cup \set{u})$. Since every edge in $\edges(H){\setminus}\edges(G)$ is between $s$ and a vertex in $C_s(H\sminus S)\cup \set{u}=C_s(G\sminus S)$, we have that $C_t=C_t(G\sminus S)$, and hence $S=N_G(C_t(G\sminus S))$. By Lemma~\ref{lem:fullComponents}, $S\in \minlsepst{G}$ where $C_s(G\sminus S)=C_s\cup \set{u}$. 
\end{proof}

\begin{replemma}{\ref{lem:inclusionCsCt}}
	\inclusionCsCt
\end{replemma}
\begin{proof}
	If $C_s(G\sminus S)\subseteq C_s(G\sminus T)$ then $N_G(C_s(G\sminus S))\subseteq C_s(G\sminus T)\cup N_G(C_s(G\sminus T))$. Since $S,T\in \minlsepst{G}$, then by Lemma~\ref{lem:fullComponents}, it holds that $S=N_G(C_s(G\sminus S))$ and $T=N_G(C_s(G\sminus T))$. Therefore, $S \subseteq C_s(G\sminus T)\cup T$. Hence, $C_s(G\sminus S)\subseteq C_s(G\sminus T) \Longrightarrow S \subseteq C_s(G\sminus T)\cup T$. 
	If $S \subseteq C_s(G\sminus T)\cup T$, then by definition, $S\cap C_t(G\sminus T)=\emptyset$. Therefore, $C_t(G\sminus T)$ is connected in $G\sminus S$. By definition, this means that $C_t(G\sminus T)\subseteq C_t(G\sminus S)$. Therefore, $N_G(C_t(G\sminus T))\subseteq C_t(G\sminus S) \cup N_G(C_t(G\sminus S))$. Since $S,T\in \minlsepst{G}$, then by Lemma~\ref{lem:fullComponents}, it holds that $S=N_G(C_t(G\sminus S))$ and $T=N_G(C_t(G\sminus T))$. Consequently, $T\subseteq S\cup C_t(G\sminus S)$. So, we have shown that
	$C_s(G\sminus S)\subseteq C_s(G\sminus T)  \Longrightarrow S\subseteq T\cup C_s(G\sminus T) \Longrightarrow T\subseteq S\cup C_t(G\sminus S)$.
	If $T\subseteq S\cup C_t(G\sminus S)$, then by definition, $T\cap C_s(G\sminus S)=\emptyset$. Therefore, $C_s(G\sminus S)$ is connected in $G\sminus T$. Consequently, $C_s(G\sminus S)\subseteq C_s(G\sminus T)$. \qed
\end{proof}

\section{Missing Proofs from Section~\ref{sec:AlgOverview}}
\label{sec:WGAppendixAlgOverview}
\begin{replemma}{\ref{lem:FsaExists}}
	\FsaExists
\end{replemma}
\begin{proof}
	By induction on $|C_s(G\sminus S)|$. If $|C_s(G\sminus S)|=|sA|$, then $C_s(G\sminus S)=sA$. By definition, $S\in \F_{sA}(G)$. Suppose the claim holds for the case where $|C_s(G\sminus S)|\leq k$ for some $k\geq |sA|$, we prove for the case where $|C_s(G\sminus S)|= k+1$. If $S\in \F_{sA}(G)$, then we are done. Otherwise, there exists a $S'\in \minlsepst{G}{\setminus}\set{S}$ where $A\subseteq C_s(G\sminus S')$ and $C_s(G\sminus S')\subseteq C_s(G\sminus S)$. Since $S'\neq S$, then $C_s(G\sminus S')\subset C_s(G\sminus S)$. Since $|C_s(G\sminus S')|< |C_s(G\sminus S)|=k+1$, then by the induction hypothesis, there exists a $T\in \F_{sA}(G)$ where $C_s(G\sminus T)\subseteq C_s(G\sminus S')\subset C_s(G\sminus S)$. \qed
\end{proof}

\begin{replemma}{\ref{lem:minSafeSepOverview1}}
	\lemminSafeSepOverview1
\end{replemma}
\begin{proof}
	Let $s\in A$, $t\in B$.	If $S\in \minlsepst{G}$ where $A\subseteq C_s(G\sminus S)$ and $B\subseteq C_t(G\sminus S)$ then clearly $S$ is a safe $A,B$-separator. 
	By Lemma~\ref{lem:fullComponents}, it holds that $S=N_G(C_s(G\sminus S))\cap N_G(C_t(G\sminus S))$.
	By Lemma~\ref{lem:simpAB}, it holds that $S$ is a minimal, safe $A,B$-separator.
	
	Now, let $S$ be a minimal, safe $A,B$-separator, where $C_A,C_B\in \cc(G\sminus S)$ contain $A$ and $B$ respectively. 
	By Lemma~\ref{lem:simpAB}, it holds that $S=N_G(C_A)\cap N_G(C_B)$. 
	By Lemma~\ref{lem:fullComponents}, $S\in \minlsepst{G}$ for every pair of vertices $s\in A$ and $t\in B$, where $C_s(G\sminus S)=C_A$ and $C_t(G\sminus S)=C_B$.\qed
\end{proof}
\section{Proofs from Section~\ref{sec:charMinlSeps}}
\label{sec:AppendixProofsEssentialMiinlstseps}
\begin{replemma}{\ref{lem:technicalLemmaCloseToA}}
	\technicalLemmaCloseToA
\end{replemma}
\begin{proof}
	Let $T\in \minlsepst{G}$ such that $A\subseteq C_s(G\sminus T)$. Since $T_s\subseteq N_G(s)$, then $T_s \subseteq T\cup C_s(G\sminus T)$. If $T_s \subseteq T$, then since $T,T_s \in \minlsepst{G}$, then $T_s=T$. But then, $A\not\subseteq C_s(G\sminus T)$; a contradiction. Therefore, for every $T\in \minlsepst{G}$ where $A\subseteq C_s(G\sminus T)$, it holds that $T_s \cap C_s(G\sminus T)\neq \emptyset$.
	
	For every $v\in T_s$, we have that $G[svA]$ is connected. Indeed, $T_s \subseteq N_G(s)$, and hence $(s,v)\in \edges(G)$. By the assumption of the lemma $T_s \subseteq N_G(a)$ for every $a\in A$. Therefore, $v\in \bigcap_{a\in A}N_G(a)$. By Corollary~\ref{corr:singleCloseSet}, there exists a unique minimal $s,t$-separator $S_v\in \minlsepst{G}$ that is close to $svA$. Let $T_s=\set{v_1,\dots,v_\ell}$, and let $S_i\in \minlsepst{G}$ denote the unique minimal $s,t$-separator that is close to $sv_iA$. We now show that for every $T\in \minlsepst{G}$ where $A\subseteq C_s(G\sminus T)$ it holds that $C_s(G\sminus S_i)\subseteq C_s(G\sminus T)$ for some $i\in \set{1,2,\dots,\ell}$. We have shown that $T_s\cap C_s(G\sminus T)\neq \emptyset$. Let $v_i \in C_s(G\sminus T)$. Therefore, $sv_iA \subseteq C_s(G\sminus T)$. Since $S_i\in \minlsepst{G}$ is the unique minimal $s,t$-separator that is close to $sv_iA$, then $C_s(G\sminus S_i) \subseteq C_s(G\sminus T)$. Since the $S_i$s are not necessarily distinct, there are at most $|T_s|$ minimal $s,t$-separators that are close to $sA$. Specifically, these are $\set{S\in \minlsepst{G}: S\subseteq N_G(sv_iA), v_i\in T_s}$. By Corollary~\ref{corr:singleCloseSet}, every $S\in \minlsepst{G}$ where $S\subseteq N_G(sv_iA)$ and $v_i\in T_s \subseteq N_G(s)$ is unique and can be found in time $O(m)$.\qed
\end{proof}

\begin{replemma}{\ref{lem:belongtoCs}}
	\belongtoCsLem
\end{replemma}
\begin{proof}
	Since $u\notin C_s(G\sminus S)$, then every path from $u$ to $s$ passes through a vertex in $S$. Now, let $T\in \minlsepst{G}{\setminus} \set{S}$. Since $S\subseteq N_G(t)$, then $S\subseteq T\cup C_t(G\sminus T)$. Therefore, every path from a vertex in $S$ to $s$ passes through a vertex in $T$. Consequently, every path from $u$ to $s$, which passes through a vertex in $S$, must also pass through a vertex in $T$. Therefore, $u\notin C_s(G\sminus T)$.
\end{proof}

\begin{replemma}{\ref{lem:lemNotCsOrCt}}
	\lemNotCsOrCt
\end{replemma}
\begin{proof}
	Let $v\in T_D$, and let $S\in \minlsepst{G}$. If $v \notin C_s(G\sminus S)$, then since $v\in T_D\subseteq T_S\subseteq N_G(s)$, then $v\in S$. 
	Therefore, $$\minlsepst{G}{\setminus}\left(\bigcup_{v\in T_D}\set{S\in \minlsepst{G}: v\in C_s(G\sminus S)}\right)=\set{S\in \minlsepst{G}: T_D\subseteq S}.$$
	To prove the claim of the lemma, we show that the complement sets are equal. 
	\begin{equation}
		\label{eq:NotCsOrCt1}
		\set{S\in \minlsepst{G}: A\not\subseteq C_s(G\sminus S)}=\set{S\in \minlsepst{G}: T_D \subseteq S}
	\end{equation}
	Let $S\in \minlsepst{G}$. Since $T_D$ is an $s,A$-separator for every $A\subseteq D$, then if $T_D \subseteq S$, then $A\not\subseteq C_s(G\sminus S)$. 
	
	For containment in the other direction, take $S\in \minlsepst{G}$ where $A\not\subseteq C_s(G\sminus S)$. Let $a\in A\subseteq D$ such that $a\notin C_s(G\sminus S)$. Since $a\in D$ then, by Lemma~\ref{lem:fullComponents}, $T_D\in \minlsep{s,a}{G}$ where $T_D \subseteq T_s\subseteq N_G(s)$. By Lemma~\ref{lem:uniqueCloseVertexSet}, $T_D$ is the unique minimal $s,a$-separator that is close to $s$ where $D=C_a(G\sminus T_D)$ (see illustration in Figure~\ref{fig:illustrationlemNotCsOrCt}). Since $C_s(G\sminus T_D)\subseteq C_s(G\sminus T)$ for every $T\in \minlsep{s,a}{G}$, then by Lemma~\ref{lem:inclusionCsCt}, it holds that $T\subseteq T_D \cup D$. 
	By Lemma~\ref{lem:uniqueCloseVertexSet}, $T_s$ is the unique minimal $s,t$-separator that is close to $s$. Therefore, $S\subseteq T_s\cup C_t(G\sminus T_s)$ for every $S\in \minlsepst{G}$. 
	
	If $a\notin C_s(G\sminus S)$, then $S\supseteq T$ for some $T\in \minlsep{s,a}{G}$. Since $T\subseteq T_D\cup D$, then we can express $T=T_1\cup T_2$ where $T_1 \eqdef T\cap T_D$ and $T_2 \eqdef T \cap D$. Likewise, since $S\subseteq T_s \cup C_t(G\sminus T_s)$, then we can write $S=S_1 \cup S_2$, where $S_1 \eqdef S\cap T_s$ and $S_2\eqdef S\cap C_t(G\sminus T_s)$. Since $S\supseteq T$, then $S_1 \cup S_2 \supseteq T_1 \cup T_2$. Since $T_2 \subseteq D$, then $T_2\cap S\subseteq D\cap (T_s\cup C_t(G\sminus T_s))=\emptyset$. Therefore, if $T_1\cup T_2\subseteq S$, then $T_2=\emptyset$. This means that $T = T_1 \subseteq T_D$. Since $T,T_D\in \minlsep{s,a}{G}$, then $T=T_D$. Therefore, if $T\subseteq S$ for some $T\in \minlsep{s,a}{G}$, then $T_D\subseteq S$.  
	So, we showed that if $S\in \minlsepst{G}$ where $a\notin C_s(G\sminus S)$ for some $a\in A$, then $T_D \subseteq S$.\qed
\end{proof}

\begin{replemma}{\ref{lem:heirarchical}}
	\lemheirarchical
\end{replemma}
\begin{proof}
	If $C_1=C_2$ the claim clearly holds, so we assume the two components are distinct.
	By definition, $N_G(C_1)\cup N_G(C_2) \subseteq T_s$. By Lemma~\ref{lem:fullComponents}, it holds that $T_s=N_G(C_t(G\sminus T_s))$. Therefore, if $C_1=C_t(G\sminus T_s)$ or $C_2=C_t(G\sminus T_s)$, then the claim clearly holds. So, we assume that $C_1,C_2\in \cc(G\sminus T_s){\setminus}\set{C_s(G\sminus T_s), C_t(G\sminus T_s)}$. 
	
	Suppose, by way of contradiction, that $N_G(C_1)\not\subseteq N_G(C_2)$ and $N_G(C_2)\not\subseteq N_G(C_1)$. Let $v_1 \in N_G(C_1){\setminus}N_G(C_2)$ and $v_2 \in N_G(C_2){\setminus}N_G(C_1)$. Also, let $u_1\in C_1$ and $u_2 \in C_2$. By our assumption, $v_1\not\in C_2 \cup N_G(C_2)$, and hence $v_1 \notin N_G[u_2]$. Likewise, $v_2\not\in C_1 \cup N_G(C_1)$, and hence $v_2 \notin N_G[u_1]$ (see illustration in Figure~\ref{fig:illustrationlemheirarchical}). Since $v_1,v_2\in T_s$, then by Lemma~\ref{lem:fullComponents}, it holds that $v_1,v_2 \in N_G(C_t(G\sminus T_s))$. Therefore, there is a $u_1,t$-path $P_{u_1,t}$ via $v_1$ such that $\nodes(P_{u_1,t})\subseteq C_1 \cup \set{v_1} \cup C_t(G\sminus T_s)$, and hence $\nodes(P_{u_1,t})\cap N_G[u_2]=\emptyset$. Likewise, there is a $u_2,t$-path $P_{u_2,t}$ via $v_2$ such that $\nodes(P_{u_2,t})\subseteq C_2 \cup \set{v_2} \cup C_t(G\sminus T_s)$, and hence $\nodes(P_{u_2,t})\cap N_G[u_1]=\emptyset$ (see illustration in Figure~\ref{fig:illustrationlemheirarchical}). Finally, since $v_1,v_2\in T_s\subseteq N_G(s)$, then there is a $u_1,u_2$-path contained entirely in $C_1\cup C_2 \cup \set{s,v_1,v_2}$. Since, by our assumption, $T_s\cap N_G[t]=\emptyset$, then this path, denoted $P_{u_1,u_2}$ (see Figure~\ref{fig:illustrationlemheirarchical}) avoids $N_G[t]$. But then, $u_1,u_2,t$ form an asteroidal triple in $G$, a contradiction (see Figure~\ref{fig:illustrationlemheirarchical2}). \qed
\end{proof}
\begin{figure}[H]
	\centering
	\begin{minipage}{0.58\textwidth}
		\centering
		\resizebox{0.7\textwidth}{!}{
			\begin{tikzpicture}
				\draw[thick] (0,0) ellipse (0.3cm and 1.3cm);
				
				\foreach \y in {-1,-0.5,0,0.5,1} {
					\fill[black] (0, \y) circle (2pt);
				}
				
				\node[fill=black, circle, inner sep=2pt, label=left:$s$] (s) at (-1,0) {};
				
				\foreach \y in {-1,-0.5,0,0.5,1} {
					\draw (s) -- (0, \y);
				}
				
				\draw[thick,green](s)--(0,-1);
				\draw[thick,green](s)--(0,0.5);
				
				\node at (0.05,-1.6) {$T_s$};
				
				\draw[thick] (3, 0) ellipse (1.5cm and 1cm);
				
				\node at (3, -1.5) {$C_t(G \sminus T_s)$};
				
				\coordinate (p1) at (1.85, 0.64);
				\coordinate (p2) at (1.59, 0.34);
				\coordinate (p3) at (1.50, 0);
				\coordinate (p4) at (1.59, -0.34);
				\coordinate (p5) at (1.85, -0.64);
				
				\draw (0, 1) -- (p1);
				\draw[thick, blue] (0, 0.5) -- (p2);
				\draw (0, 0) -- (p3);
				\draw (0, -0.5) -- (p4);
				\draw[thick, red] (0, -1) -- (p5);
				
				\node at (0, -0.81) {\footnotesize{$v_1$}};
				
				\node at (0, 0.66) {\footnotesize{$v_2$}};
				
				\draw[thick] (-1,-1.8) ellipse (0.5cm and 0.3cm);
				\node at (-1.7,-1.8) {$C_1$};
				
				\draw[thick] (1.2,-2) ellipse (0.4cm and 0.5cm);
				\node at (0.78,-2.5) {$C_2$};
				
				\coordinate (q1) at (-0.51, -1.74);
				\coordinate (q2) at (-0.56, -1.66);
				\coordinate (q3) at (-0.69, -1.57);
				
				\draw[thick, bend right,red] (0, -1) to (q1);
				\draw[bend right] (0, -0.5) to (q2);
				\draw[bend right] (0, 0) to (q3);
				
				\coordinate (x1) at (1.48,-1.64);
				\coordinate (x2) at (1.23,-1.50);
				
				\draw[bend left] (0, -0.5) to (x2);
				\draw[thick, blue,bend left] (0, 0.5) to (x1);

				\coordinate (u1) at (-1,-1.75);
				\coordinate (u2) at (1.2,-1.9);
				
				\node[fill=black, circle, inner sep=0.6pt] at (u1) {};
				\node at (-1.17,-1.75) {\footnotesize{$u_1$}};
				
				\node[fill=black, circle, inner sep=0.6pt] at (u2) {};
				\node at (1.2,-2.05) {\footnotesize{$u_2$}};
				
				\node[fill=black, circle, inner sep=2pt, label=right:$t$] at (4.5,0) {};
				
				\draw[thick, red, decorate, decoration={snake, amplitude=0.5mm}] (u1) -- (q1);
				
				\draw[thick, blue, decorate, decoration={snake, amplitude=0.5mm}] (u2) -- (x1);
				
				\draw[thick, red, decorate, decoration={snake, amplitude=0.5mm}] (p5) -- (4.5,0);
				
				\draw[thick, blue, decorate, decoration={snake, amplitude=0.5mm}] (p2) -- (4.5,0);
				
				\node[red, rotate=22] at (3.2, -0.55) {\footnotesize{$P_{u_1,t}$}};
				\node[blue, rotate=-5] at (3.4, 0.37) {\footnotesize{$P_{u_2,t}$}};
				
				\node[green] at (-0.95, -0.5) {\footnotesize{$P_{u_1,u_2}$}};
				
			\end{tikzpicture}
		}
		\caption{Illustration--Lemma~\ref{lem:heirarchical}.\label{fig:illustrationlemheirarchical}}
	\end{minipage}
	\hfill
	\begin{minipage}{0.4\textwidth}
		\resizebox{0.77\textwidth}{!}{
			\begin{tikzpicture}
				\coordinate (u1) at (0, 0);
				\coordinate (u2) at (4, 0);
				\coordinate (t) at (2, 3);
				
				\draw[thick, green, decorate, decoration={snake, amplitude=0.5mm}] (u1) -- (u2);
				\node[green] at (2,-0.3) {\footnotesize{$P_{u_1,u_2}$}};
				
				\draw[thick, red, decorate, decoration={snake, amplitude=0.5mm}] (u1) -- (t);
				\node[blue, rotate=-45] at (3.4,1.5) {\footnotesize{$P_{u_2,t}$}};
				
				\draw[thick, blue, decorate, decoration={snake, amplitude=0.5mm}] (u2) -- (t);
				\node[red, rotate=45] at (0.58,1.5) {\footnotesize{$P_{u_1,t}$}};
				
				\fill[black] (u1) circle (2pt) node[left] {$u_1$};
				\fill[black] (u2) circle (2pt) node[right] {$u_2$};
				\fill[black] (t) circle (2pt) node[above] {$t$};
			\end{tikzpicture}
		}
		\caption{Illustration--Lemma~\ref{lem:heirarchical}.\label{fig:illustrationlemheirarchical2}}
	\end{minipage}
\end{figure}

\begin{repcorollary}{\ref{corr:heirarchical2}}
	\corrheirarchical2
\end{repcorollary}
\begin{proof}
	Since $\emptyset \subset A \subseteq \nodes(G){\setminus}(C_s(G\sminus T_s)\cup T_s\cup C_t(G\sminus T_s))$, then $\ell \geq 1$.
	Assume wlog that $|N_G(C_1)|\leq |N_G(C_2)|\leq \cdots \leq |N_G(C_\ell)|$. Since $G$ is AT-free, and $T_s\cap N_G[t]=\emptyset$, then by Lemma~\ref{lem:heirarchical}, we have that $N_G(C_1)\subseteq N_G(C_2)\subseteq \cdots \subseteq N_G(C_\ell)$. Therefore, $S^*=N_G(C_1)$. 
	
	Let $S\in \minlsepst{G}$ where $v\in C_s(G\sminus S)$ for some $v\in S^*$. Since $v\in S^*\eqdef \bigcap_{i=1}^\ell N_G(C_i)$, then by Lemma~\ref{lem:lemNotCsOrCt}, we have that $S\in \minlsepst{G}$ where $C_i \subseteq C_s(G\sminus S)$ for every $i\in \set{1,2,\dots,\ell}$. Since $A\subseteq \bigcup_{i=1}^\ell C_i$, then $A\subseteq C_s(G\sminus S)$.

	Suppose, by way of contradiction, that there exists an $S\in \minlsepst{G}$ where $A\subseteq C_s(G\sminus S)$ and $C_s(G\sminus S)\cap S^*=C_s(G\sminus S)\cap N_G(C_1)=\emptyset$. Since $S^* \subseteq T_s \subseteq N_G(s)$, it means that $S^* \subseteq S$. Since $S^*$ is, by definition, an $s,C_1$-separator and $S^*\subseteq S$, then $S$ is an $s,C_1$-separator. Therefore, $C_1\cap C_s(G\sminus S)=\emptyset$. In particular, $(C_1\cap A)\cap C_s(G\sminus S)=\emptyset$. Since $C_1 \cap A\neq \emptyset$, then $A\not\subseteq C_s(G\sminus S)$; a contradiction.
\end{proof}
\section{Missing Details from Section~\ref{sec:MainThm}: Pseudocode and Runtime Analysis of the $\algname{CloseTo}$ Procedure}
\label{sec:CloseToPseudocode}
In this Section, we describe the algorithm $\algname{CloseTo}$ (Figure~\ref{fig:closeTo}) that receives as input a weighted, AT-free graph $G$, two distinct vertices $s,t\in \nodes(G)$, and a subset $A\subseteq \nodes(G)$, and returns the set $\F_{sA}(G)$ of minimal $s,t$-separators that are close to $sA$ (Definition~\ref{def:closeToA}).

Before describing algorithm $\algname{CloseTo}$, we complete the proof of Claim 5, as part of the proof of Theorem~\ref{thm:mainThmATFree}, detailed in Section~\ref{sec:MainThm}.

\subsection{Completing the proof of Theorem~\ref{thm:mainThmATFree}: Claim 5.}
\noindent {\bf Claim 5}: $\F_{sD}(M)=\F_{sA}(G)$.\newline
\noindent{\bf Proof}:
Let $S\in \F_{sA}(G)$. By definition, $\F_{sA}(G)\subseteq \set{S\in \minlsepst{G}: A\subseteq C_s(G\sminus S)}$. From eq.~\eqref{eq:MainThm1} and~\eqref{eq:MainThm2}, we have that $S\in \minlsepst{M}$, where $C_s(G\sminus S)=C_s(M\sminus S)\cup C_s(G\sminus Q_s)$. Since $A\subseteq C_s(G\sminus S)$, then $D=A{\setminus}C_s(G\sminus Q_s)\subseteq C_s(M\sminus S)$. Suppose, by way of contradiction, that $S\notin \F_{sD}(M)$. Since $D\subseteq C_s(M\sminus S)$, it means that there exists an $S'\in \minlsepst{M}$ where $D\subseteq C_s(M\sminus S')\subset C_s(M\sminus S)$. By~\eqref{eq:MainThm2}, we have that $S'\in \minlsepst{G}$ where $C_s(G\sminus S')=C_s(M\sminus S')\cup C_s(G\sminus Q_s)\subset C_s(M\sminus S)\cup C_s(G\sminus Q_s)=C_s(G\sminus S)$. But then, $A\subseteq C_s(G\sminus S')\subset C_s(G \sminus S)$, contradicting the assumption that $S\in \F_{sA}(G)$. Therefore, $\F_{sA}(G)\subseteq \F_{sD}(M)$.

Now, let $S\in \F_{sD}(M)$. From~\eqref{eq:MainThm2}, we have that $S\in \minlsepst{G}$ where $C_s(G\sminus S)=C_s(M\sminus S)\cup C_s(G\sminus Q_s)$. Since $D\subseteq C_s(M\sminus S)$, we have that $A\subseteq C_s(G\sminus S)$. If $S\notin \F_{sA}(G)$, then there exists an $S'\in \minlsepst{G}$ where $A\subseteq C_s(G\sminus S') \subset C_s(G\sminus S)$. From~\eqref{eq:MainThm1} and~\eqref{eq:MainThm2}, we have that $S'\in \minlsepst{M}$, where $C_s(G\sminus S')=C_s(M\sminus S')\cup C_s(G\sminus Q_s)$. Therefore, we have that $A\subseteq C_s(G\sminus S')=C_s(M\sminus S')\cup C_s(G\sminus Q_s)\subset C_s(M\sminus S)\cup C_s(G\sminus Q_s)=C_s(G\sminus S)$. In particular, this means that $D\subseteq C_s(M\sminus S')$ and $C_s(M\sminus S')\subset C_s(M\sminus S)$, contradicting the assumption that $S\in \F_{sD}(M)$.

\subsection{Pseudocode and Runtime Analysis of the $\algname{CloseTo}$ Procedure}
Algorithm $\algname{CloseTo}$ (Figure~\ref{fig:closeTo})  receives as input a weighted, AT-free graph $G$, two distinct vertices $s,t\in \nodes(G)$, and a subset $A\subseteq \nodes(G)$, and returns the set $\F_{sA}(G)$ of minimal $s,t$-separators that are close to $sA$ (Definition~\ref{def:closeToA}).

If $\F_{sA}\neq \emptyset$, then $sA\cap N_G[t]=\emptyset$ (line~\ref{line:returnnull}). If $S\in \F_{sA}(G)$, then $L\eqdef N_G(sA)\cap N_G(t)\subseteq S$. Therefore, the algorithm processes $G'\eqdef G\sminus L$, which is also AT-free (line~\ref{line:L}).  
Let $T_t,T_s\in \minlsepst{G'}$, where $T_t\subseteq N_{G'}(t)$ and $T_s\subseteq N_{G'}(s)$ which, by Lemma~\ref{lem:uniqueCloseVertexSet}, are unique and can be computed in time $O(m)$ (lines~\ref{line:Tt} and~\ref{line:Ts}). 
If $sA\not\subseteq C_s(G'\sminus T_t)$, then by Lemma~\ref{lem:belongtoCs}, it holds that $A\not\subseteq C_s(G'\sminus S)$ for every $S\in \minlsepst{G'}$. Hence, $\F_{sA}(G)=\emptyset$ is returned in line~\ref{line:returnnull2}. 
If $sA \subseteq C_s(G'\sminus T_s)$, then since $T_s\subseteq N_{G'}(s)$, then $T_s \subseteq S\cup C_s(G'\sminus S)$ for every $S\in \minlsepst{G'}$. By Lemma~\ref{lem:inclusionCsCt}, $sA\subseteq C_s(G'\sminus T_s)\subseteq C_s(G'\sminus S)$ for every $S\in \minlsepst{G'}$. By Definition~\ref{def:closeToA} of minimal separator close to $sA$, we get that $T_s$ is the unique minimal $s,t$-separator that is close to $sA$. Therefore, $\set{(T_s\cup L)}$ it is returned in line~\ref{line:inCsTs}.

By Corollary~\ref{corr:heirarchical2} (see also~\eqref{eq:mainThm22}), we have that  $\F_{sA}(G')=\bigcup_{v\in S^*}\F_{sA_v}(G')$ where $A_v\eqdef A_1\cup \set{v}$, and where $A_1\eqdef A\cap(C_s(G\sminus T_s)\cup T_s \cup C_t(G\sminus T_s))$, and $S^*$ is computed in line~\ref{line:sstar}.
Therefore, the algorithm iterates over all $v\in S^*$ in lines~\ref{line:loopStart}-\ref{line:loopEnd}, and computes $\F_{sA_v}(G')$ for every $v\in S^*$. In lines~\ref{line:getH} and~\ref{line:getS1}, the algorithm computes $S_1\in \minlsepst{H}=\minlsep{sA,t}{G'}$ according to Lemma~\ref{lem:MinlsASep}. If $A_v \subseteq C_s(H\sminus S_1)$, then $\F_{sA_v}(G)=\set{(S_1\cup L)}$ (Case 1 in the proof of Theorem~\ref{thm:mainThmATFree}). Otherwise, the algorithm generates the graph $M$ where $Q_s \subseteq N_M(s)$ in line~\ref{line:genM}. By~\eqref{eq:MainThm1} and~\eqref{eq:MainThm2}, we have that $\F_{sA_v}(G')=\F_{sD_v}(M)\subseteq \minlsepst{M}$, where $D_v\eqdef A_v{\setminus}C_s(G\sminus Q_s)$. By Lemma~\ref{lem:technicalLemmaCloseToA}, there are at most $|Q_s|$ minimal $s,t$-separators that are close to $sD_v$ in $M$; one for every $w\in Q_s$ that are generated in the loop in lines~\ref{line:loopwStart}-\ref{line:loopwEnd}.     

The pseudocode of Figure~\ref{fig:closeTo} presents the algorithm in the case where $S^*\neq \emptyset$. If $S^*=\emptyset$, then $A\subseteq C_s(G'\sminus T_s)\cup T_s \cup C_t(G'\sminus T_s)$, and the algorithm will execute the pseudocode in lines~\ref{line:loopStartp1}-\ref{line:loopwEnd} just once where $A_v=A$.

\noindent{\bf Runtime.} The runtime of the procedure is $\max\set{1,|S^*|}\cdot O(n \cdot m)$, and hence $O(n^2m)$.

\renewcommand{\algorithmicrequire}{\textbf{Input: }}
\renewcommand{\algorithmicensure}{\textbf{Output: }}
\def\F{\mathcal{F}}
\begin{algserieswide}
	{H}{Algorithm for returning the minimal $s,t$-separators that are close to $sA$ according to Definition~\ref{def:closeToA}. \label{fig:closeTo}}
	\begin{insidealgwide}{CloseTo}{$G$, $s$, $t$, $A$}
		\REQUIRE{AT-free graph $G$, $s,t\in \nodes(G)$, and $A\subseteq \nodes(G){\setminus}\set{s,t}$.}
		\ENSURE{$\F_{sA}(G)$ (Definition~\ref{def:closeToA})}.
		\IF{$sA\cap N_G[t]\neq \emptyset$}
		\RETURN $\emptyset$ \label{line:returnnull}
		\ENDIF
		\STATE $L\eqdef N_G(sA) \cap N_G(t)$
		\STATE $G'\eqdef G\sminus L$ \label{line:L}
		\STATE Compute $T_t\in \minlsepst{G'}$ where $T_t\subseteq N_{G'}(t)$ \COMMENT{Lemma~\ref{lem:uniqueCloseVertexSet}} \label{line:Tt}
		\IF{$sA\not\subseteq C_s(G'\sminus T_t)$}
		\RETURN $\emptyset$ \label{line:returnnull2} \COMMENT{Lemma~\ref{lem:belongtoCs}}
		\ENDIF
		\STATE Compute $T_s\in \minlsepst{G'}$ where $T_s\subseteq N_{G'}(s)$ \COMMENT{Lemma~\ref{lem:uniqueCloseVertexSet}} \label{line:Ts}
		\IF{$sA \subseteq C_s(G'\sminus T_s)$}
		\RETURN $\set{(T_s\cup L)}$ \label{line:inCsTs}
		\ENDIF
		\STATE Let $\set{C_1,\dots,C_\ell}\eqdef \set{C\in \cc(G'\sminus T_s): A\cap C\neq \emptyset, s\notin C, t\notin C}$
		\STATE $S^* \eqdef \bigcap_{i=1}^\ell N_{G'}(C_i)$ \label{line:sstar}
		\STATE $\F \gets \emptyset$
		\FORALL{$v\in S^*$} \label{line:loopStart}
		\STATE $A_v \eqdef (A\cap(C_s(G\sminus T_s)\cup T_s \cup C_t(G\sminus T_s)))\cup \set{v}$  \label{line:loopStartp1}
		\STATE Let $H$ be the graph where $\nodes(H)\eqdef \nodes(G')$ and $\edges(H)\eqdef \edges(G')\cup\set{(s,z): z\in N_{G'}[A_v]}$ \label{line:getH}
		\STATE Let $S_1\in \minlsepst{H}$ where $S_1\subseteq N_H(s)$ \COMMENT{By Lemma~\ref{lem:MinlsASep}, $S_1\in \minlsep{sA,t}{G'}$} \label{line:getS1}
		\IF{$A_v \subseteq C_s(H\sminus S_1)$}
		\STATE $\F \gets \F \cup \set{(S_1\cup L)}$
		\ELSE
		\STATE $Q_s\gets N_{G'}(C_s(H\sminus S_1))\cap S_1$ \COMMENT{$Q_s\in \minlsepst{G'}$. By Claim 4, $Q_s \subseteq N_{G'}(a)$, for every $a\in A_v{\setminus}C_s$}
		\STATE Let $M$ be the graph that results from $G'$ by contracting $C_s(G'\sminus Q_s)$ to vertex $s$ \label{line:genM} \COMMENT{See~\eqref{eq:MainThm1} and~\eqref{eq:MainThm2}.}
		\FORALL{$w \in Q_s$} \label{line:loopwStart}
		\STATE Let $M_w$ be the graph that results from $M$ by contracting $(s,w)$ to $s$. \COMMENT{Lemma~\ref{lem:contract}}
		\STATE Let $T_w \in \minlsepst{M_w}$ where $T_w \subseteq N_{M_w}(s)$
		\STATE $\F \gets \F \cup \set{(T_w \cup L)}$ \label{line:loopwEnd}
		\ENDFOR
		\ENDIF
		\ENDFOR \label{line:loopEnd}
		\RETURN $\F$
	\end{insidealgwide}
\end{algserieswide}

\section{Hardness of \textsc{Min-Safe Separator}}
\label{sec:hardnessOfMinSafeSep}
\eat{
\subsection{Minimal Separators and the 2-Disjoint Connected Subgraphs Problem}
\label{sec:inducedDisjointPath}
}
\eat{
Let $A$ and $B$ be two disjoint, non-adjacent subsets of $\nodes(G)$. A vertex-set $X\subseteq \nodes(G)$ is called an \e{$A,B$-separator} if, in the graph $G[\nodes(G)\setminus X]$, there is no path between $A$ and $B$. We say that $X$ is a minimal $A,B$-separator if no proper subset of $X$ has this property. 
We denote by $\minlsep{A,B}{G}$  the set of minimal $A,B$-separators in $G$. We say that an $A,B$-separator $X$ is 
\e{safe} if there are two distinct, connected components $C_A,C_B\in \cc_G(X)$, where $A\subseteq C_A$ and $B\subseteq C_B$.
We say that a subset $X\subseteq \nodes(G)\setminus (A\cup B)$ is a minimum $A,B$-separator if, for every $A,B$-separator $S$, it holds that $|X|\leq |S|$. We denote by $\minsep_{A,B}(G)$ the set of minimum $A,B$-separators, and by $\kappa_{A,B}(G)$ the size of a minimum $A,B$-separator.
}
The \textsc{2-disjoint connected subgraphs} problem is an intensively studied problem defined as follows.
The input is an undirected graph $G$ together with two disjoint subsets of vertices $A,B\subseteq \nodes(G)$. The goal is to decide whether there exist two disjoint subsets $A_1,B_1\subseteq \nodes(G)$, such that $A\subseteq A_1$, $B\subseteq B_1$, and $G[A_1]$ and $G[B_1]$ are connected. 
For two disjoint subsets $A,B\subseteq \nodes(G)$ in an undirected graph $G$, we denote by $\texttt{2Dis}(G,A,B)$ the instance of the \textsc{2-Disjoint Connected Subgraph} problem in $G$ with vertex-subsets $A$ and $B$.

The \textsc{2-Disjoint Connected Subgraphs} problem is NP-complete~\cite{van_t_hof_partitioning_2009}, and remains so even if one of the input vertex-sets contains only two vertices, or if the input graph contains a $P_4$~\cite{van_t_hof_partitioning_2009}. Motivated by an application in computational-geometry, Gray et al.~\cite{gray_removing_2012} show that the \textsc{2-Disjoint Connected Subgraphs} problem is NP-complete even for the class of planar graphs.
A na\"ive brute-force algorithm that tries all $2$-partitions of the vertices in $\nodes(G)\setminus (A\cup B)$ runs in time $O(2^nn^{O(1)})$.
Cygan et al.~\cite{cygan_solving_2014} were the first to present an exponential time algorithm for general graphs that is faster than the trivial $O(2^nn^{O(1)})$ algorithm, and runs in time $O^*(1.933^n)$ (i.e., excluding poly-logarithmic terms). This result was later improved by Telle and Villanger~\cite{DBLP:conf/wg/TelleV13}, that presented an enumeration-based algorithm that runs in time $O^*(1.7804^n)$. 

Restricting the input to the \textsc{2-Disjoint Connected Subgraphs} problem to special graph classes has been the focal point of previous research efforts. This approach has led to the discovery of islands of tractability and improved our understanding of its difficulty.
\eat{For this reason, much previous work on the \textsc{2-Disjoint Connected Subgraphs} problem focused on restricting the input graph.}
For example, in~\cite{van_t_hof_partitioning_2009}, the authors presented an algorithm that runs in polynomial time on co-graphs, and in time
$O((2-\varepsilon(\ell))^n)$ for $P_\ell$-free graphs (i.e., graphs that do not contain an induced path of size $\ell$). In subsequent work~\cite{paulusma_partitioning_2011}, the authors show that the \textsc{2-Disjoint Connected Subgraphs} problem can be solved in time $O(1.2501^n)$ on $P_6$-free graphs. More recently, Kern et al.~\cite{kern_disjoint_2022} studied the \textsc{2-Disjoint Connected Subgraphs} problem on $H$-free graphs (i.e., all graphs that do not contain the graph $H$ as an induced subgraph).
Golovach, Kratsch, and Paulusma show that \textsc{2-Disjoint Connected Subgraphs} can be solved in polynomial time in AT-free graphs~\cite{DBLP:journals/tcs/GolovachKP13}. However, their algorithms has a prohibitive runtime of $O(n^{15})$.

Deciding whether a safe separator exists remains NP-hard even if $|A|=|B|=2$. We show this by reduction from
the \textsc{Induced disjoint paths} problem. The input to this problem is an undirected graph $G$ and a collection of $k$ vertex pairs $\set{(s_1,t_1),\dots,(s_k,t_k)}$ where $s_i\neq t_i$ and $k\geq 2$. The goal is to determine whether $G$ has a set of $k$ paths that are mutually induced (i.e., they have neither common vertices nor adjacent vertices). The \textsc{Induced disjoint paths} problem remains NP-hard even if $k=2$~\cite{BIENSTOCK199185}. However, when $G$ is a planar graph and $k=2$, the problem can be solved in polynomial time, as shown by 
Kawarabayashi and Kobayashi~\cite{DBLP:conf/ipco/KawarabayashiK08}.

\eat{
	\begin{theorem}
		\label{thm:reduction2Diss}
		Let $A,B \subseteq \nodes(G)$ be disjoint. Let $a\in A$ and $b\in B$. There exists a solution to the instance $\texttt{2Dis}(G,A,B)$ of the 2-Disjoint Connected Subgraph problem 
		if and only if there exists a minimal $ab$-separator $S\in \minlsep{ab}{G}$, such that $S$ is an $AB$-partitioning.
	\end{theorem}
}
\begin{theorem}
	\label{thm:reduction2Diss}
	\sc{Min Safe Separator} is NP-Hard.
\end{theorem}
\begin{proof}
	We prove by reduction from the \textsc{2-Disjoint Connected Subgraph} problem.
	Given the instance $\texttt{2Dis}(G,A,B)$, create the graph $G'$ by subdividing every edge in $G$. In other words, we replace every edge $(u,v)\in \edges(G)$ with the two-path $(u,e_{uv},v)$ in $G'$. Now, let $A_1,B_1\subseteq \nodes(G)$ be a solution to the instance $\texttt{2Dis}(G,A,B)$. That is, $A\subseteq A_1$, $B\subseteq B_1$ and $G[A_1]$ and $G[B_1]$ are connected. Let $A_1'\subseteq \nodes(G')$ be $A_1$ plus the set of all vertices in $G'$ that correspond to edges in $G[A_1]$. Similarly, define $B_1'$ to be $B_1$ plus the set of all vertices in $G'$ that correspond to edges in $G[B_1]$. Since $G[A_1]$ and $G[B_1]$ are connected, then $G'[A_1']$ and $G'[B_1']$ are connected, and contain the sets $A$ and $B$, respectively. Clearly, $\nodes(G')\setminus (A_1'\cup B_1')$ is a safe $A,B$-separator in $G'$.\eat{ Therefore, some subset $S\subseteq \nodes(G')\setminus (A_1'\cup B_1')$ is a minimal $a,b$-separator in $G'$ that is also an  $A,B$-partitioning.}
	
	Now, let $S\subseteq \nodes(G')$ be a safe $A,B$-separator in $G'$, and let $C_A(G'\sminus S)$ and $C_B(G'\sminus S)$ be the connected components of $\cc(G'\sminus S)$ that contain $A$ and $B$, respectively.
	Let $A_1$ and $B_1$ be the vertices in $C_A(G'\sminus S)\cap \nodes(G)$ and $C_B(G'\sminus S)\cap \nodes(G)$ respectively, that correspond to the vertices of $G$ (i.e., drop the vertices of $G'$ that correspond to edges of $G$). Then $A_1$ and $B_1$ are disjoint connected vertex-sets of $\nodes(G)$, that contain $A$ and $B$, respectively, and hence a solution to $\texttt{2Dis}(G,A,B)$.\qed
	\eat{
		Now, let $A_1,B_1\subseteq \nodes(G)$ be a solution to the instance $\texttt{2Dis}(G,A,B)$. That is, $A\subseteq A_1$, $B\subseteq B_1$ and $G[A_1]$ and $G[B_1]$ are connected.
		Define $A_1'\eqdef A_1$, and repeat the following procedure. While $N_G(A_1')\setminus N_G[B_1]\neq \emptyset$, take $v\in N_G(A_1')\setminus N_G[B_1]$ and add it to $A_1'$ (i.e., $A_1' \gets A_1' \cup \set{v}$). Since $G[A_1]$ is connected, then at the end of the process $G[A_1']$  is connected as well. Also, at the end of this process $N_G(A_1')\subseteq N_G(B_1)$. We claim that $S\eqdef N_G(A_1')$ is a minimal $ab$-separator of $G$. Since $S=N_G(A_1')$ where $a\in A_1'$ and $G[A_1']$ is connected, then $S$ separates $a$ from $b$ in $G$. On the other hand, since $S\subseteq N_G[B_1]$, then for every $w\in S$, there is a path from $A_1'$ to $B_1$, and hence from $a\in A_1'$ to $b\in B_1$ in $G-(S-w)$ for any $w\in S$. Therefore, $S\in \minlsep{a,b}{G}$. Define $B_1'$ to be the connected components that contains $b$ in $G-S$. Since $A\subseteq A_1 \subseteq A_1'$, and $B \subseteq B_1 \subseteq B_1'$, then $S$ is an $AB$-partitioning.
	}
\end{proof}
\eat{An immediate consequence of Theorem~\ref{thm:reduction2Diss} is that the problem of deciding whether there exists a separator that is a partitioning is NP-hard. }
We show that \textsc{Min Safe Separator} is NP-hard even if $|A|=|B|=2$. This is done by reduction from the \textsc{Induced Disjoint Paths} Problem~\cite{DBLP:conf/ipco/KawarabayashiK08}.
A set of paths $P_1,\dots,P_k$ are said to be \e{mutually induced}~\cite{DBLP:conf/ipco/KawarabayashiK08} if they have neither common vertices, nor adjacent vertices, for every pair of distinct paths $P_i,P_j$.
\begin{citeddefinition}{\cite{DBLP:conf/ipco/KawarabayashiK08}, \sc{Induced Disjoint Paths} Problem}
	\label{def:DPP}	
	Let $G$ be an undirected graph, and $\set{(s_1,t_1),\dots,(s_k,t_k)}$ a collection of vertex pairs where, for all $i\in [1,k]$, $s_i\neq t_i$. The problem is to decide whether $G$ has a set of $k$ mutually induced paths $P_1,\dots,P_k$ such that $P_i$ is an $(s_i,t_i)$ path for $i\in [1,k]$.
\end{citeddefinition}
\begin{citedtheorem}{\cite{DBLP:conf/ipco/KawarabayashiK08}}
	\label{thm:DPP}
	{\sc{Induced Disjoint Paths}} is NP-hard when $k=2$ and $G$ is a general undirected graph. \eat{, and solvable in polynomial time when $G$ is a planar graph, and $k$ is fixed.}
\end{citedtheorem}

\begin{theorem}
	\label{thm:ABHard}
	{\sc{Min Safe Separator}} is NP-Hard when each input vertex-set contains exactly two vertices (i.e., $|A|=|B|=2$).
\end{theorem}
\begin{proof}
	We prove by reduction from \textsc{Induced Disjoint Paths}.
	Consider an instance of the induced disjoint path problem where $k=2$, and let $\set{(s_1,t_1),(s_2,t_2)}$ be the pair of non-adjacent, disjoint vertex-pairs. Define $A=\set{s_1,t_1}$ and $B=\set{s_2,t_2}$.
	Let $P_1$ and $P_2$ be mutually induced paths between $s_1$ and $t_1$, and $s_2$ and $t_2$, respectively.
	That is, the pair $P_1$, $P_2$ is a solution to \textsc{Induced Disjoint Paths}. Then $\nodes(G)\setminus (\nodes(P_1)\cup \nodes(P_2))$ is a safe $A,B$-separator.
	
	Now, suppose that $S\subseteq \nodes(G)$ is a safe $A,B$-separator of $G$.
	By definition, there exist two non-adjacent, disjoint, connected components $Z_A,Z_B\in \cc(G\sminus S)$ where $Z_A \supseteq A$ and $Z_B \supseteq B$. Since $G[Z_A]$ ($G[Z_B]$) are connected, then $G[Z_A]$ contains a path $P_1$ from $s_1$ to $t_1$. Likewise, $G[Z_B]$ contains a path $P_2$ from $s_2$ to $t_2$.
	Since $Z_A$ and $Z_B$ are disjoint and non-adjacent, $\nodes(P_1)\subseteq Z_A$, and $\nodes(P_2)\subseteq Z_B$, then the pair of paths $P_1,P_2$ is a solution to \textsc{Induced disjoint paths}. \qed
\end{proof}

\end{document}